\documentclass[11pt]{article}
\usepackage{ifthen}
\usepackage{amsmath, amssymb}

\usepackage{amsthm}

\usepackage{ifpdf}

\usepackage{fullpage}
\usepackage[utf8]{inputenc}

\usepackage{algorithm,algpseudocode}
\algtext*{EndWhile} 
\algtext*{EndIf} 
\algtext*{EndFor} 
\usepackage{amsfonts}
\usepackage{mathrsfs}
\usepackage{comment}
\usepackage{mathtools}

\usepackage{enumitem}

\usepackage{xcolor}

\usepackage{mdframed}
\usepackage{tikz}

\usepackage{bbm}
\usepackage{dsfont}
\usepackage{hyperref}
\usepackage{subcaption}

\newtheorem{oq}{Open Question}

\theoremstyle{plain}
\newtheorem{theorem}{Theorem}[section]
\newtheorem{lemma}[theorem]{Lemma}

\newtheorem{observation}[theorem]{Observation}
\newtheorem{question}[theorem]{Question}
\theoremstyle{definition}

\newcommand{\depth}{\textsf{depth}}
\newcommand{\poly}{\text{poly}}

\newcommand{\parent}{\textsf{parent}}
\newcommand{\pos}{\mathtt{pos}}

\newcommand{\mpnote}[1]{{\color{blue}[Merav: #1]}}

\usepackage{cleveref}
\begin{document}

	\title{New Distributed Interactive Proofs for Planarity: \\ A Matter of Left and Right\thanks{This project is partially funded by the European Research Council (ERC) under the European Union’s Horizon 2020 research and innovation programme, grant agreement No.\ 949083}}
	\author{Yuval Gil\footnote{Weizmann Institute of Science. yuval.gil@weizmann.ac.il} \and Merav Parter \footnote{Weizmann Institute of Science. merav.parter@weizmann.ac.il}
	}
	\date{}
	
	\maketitle	
	
	\begin{abstract}
	We provide new distributed interactive proofs (DIP) for planarity and related graph families. The notion of a \emph{distributed interactive proof} (DIP) was introduced by Kol, Oshman, and Saxena (PODC 2018). In this setting, the verifier consists of $n$ nodes connected by a communication graph $G$. The prover is a single entity that communicates with all nodes by short messages. The goal is to verify that the graph $G$ satisfies a certain property (e.g., planarity) in a small number of rounds, and with a small communication bound, denoted as the \emph{proof size}.
	
	Prior work by Naor, Parter and Yogev (SODA 2020) presented a DIP for planarity that uses three interaction rounds and a proof size of $O(\log n)$. Feuilloley et al.\ (PODC 2020) showed that the same can be achieved  with a single interaction round and without randomization, by providing a proof labeling scheme with a proof size of $O(\log n)$. In a subsequent work, Bousquet, Feuilloley, and Pierron (OPODIS 2021) achieved the same bound for related graph families such as outerplanarity, series-parallel graphs, and graphs of treewidth at most $2$. In this work, we design new DIPs that use exponentially shorter proofs compared to the state-of-the-art bounds. Our main results are:
	
	\begin{itemize}
		\item There is a $5$-round protocol with $O(\log\log n)$ proof size for outerplanarity. 
		\item There is a $5$-round protocol with $O(\log\log n)$ proof size for verifying embedded planarity and $O(\log\log n+\log \Delta)$ proof size for general planar graphs, where $\Delta$ is the maximum degree in the graph. In the former setting, it is assumed that an embedding of the graph is given (e.g., each node holds a clockwise orientation of its neighbors) and the goal is to verify that it is a valid planar embedding. The latter result should be compared with the non-interactive setting for which there is lower bound of $\Omega(\log n)$ bits for graphs with $\Delta=O(1)$ by Feuilloley et al. (PODC 2020).
		\item  The non-interactive deterministic lower bound of $\Omega(\log n)$ bits by Feuilloley et al. (PODC 2020) can be extended to hold even if the verifier is randomized. Moreover, the lower bound holds even with the assumption that the verifier's randomness comes in the form of an unbounded random string \emph{shared} among the nodes.
	\end{itemize}
	We also show that our DIPs can be extended to protocols with similar bounds for verifying series-parallel graphs and graphs with tree-width at most $2$. Perhaps surprisingly, our results demonstrate that the key technical barrier for obtaining $o(\log\log n)$ labels for all our problems is a basic sorting verification task in which all nodes are embedded on an oriented path $P \subseteq G$ and it is desired for each node to distinguish between its left and right $G$-neighbors. 
\end{abstract}
	\newpage
	\tableofcontents
	\newpage
	\section{Introduction}

Planarity is a fundamental graph property that has been widely studied due to its rich combinatorial structure and numerous algorithmic applications. While in the centralized setting, the task of verifying if a given graph is planar can be done in linear time \cite{hopcroft1974efficient}, in the distributed setting the running time depends linearly on the diameter of the graph \cite{GhaffariH16}. The non-local nature of planarity motivates the use of a powerful, but potentially untrusted, \emph{prover} that can aid the distributed verification by providing each node a short string of advice, a.k.a. a \emph{proof label}. The nodes then engage in brief communication to collectively determine whether to accept or reject the provided proof. This framework has been formalized into \emph{proof labeling schemes} by Korman, Kutten and Peleg \cite{KormanKP10}. In this work we focus on the interactive extension of this model to \emph{distributed interactive proofs} (DIP) as proposed by 
Kol, Oshman, and Saxena \cite{kol2018interactive}. In this setting, the nodes are allowed to interact with the prover through multiple rounds of communication. The key complexity measures are the number of interaction rounds and the total proof size.

The first evidence of the power of such proof systems for certifying planarity was provided by Naor, Parter, and Yogev \cite{NaorPY20}. Their result for planarity was in fact implied by a more general machinery that translates any (centralized) computation in $O(n)$ time -- such as, the centralized planarity verification of \cite{hopcroft1974efficient} -- into a three-round distributed interactive protocol with $O(\log n)$ proof size. 
Subsequent work by Feuilloley et al. \cite{FeuilloleyFMRRT21,FeuilloleyF0RRT23} demonstrated that the same proof size could be achieved with just a single interaction round, effectively reducing to the classical proof labeling scheme setting. Their work is accompanied by a matching lower bound of $\Omega(\log n)$ bits, that holds already for graphs with maximum degree of $O(1)$. 
These developments bring us back to the fundamental question of whether interaction truly provides an advantage in certifying planarity.

\begin{question}\label{q:1}
What is the power of distributed interactive proofs for certifying planarity? 
\end{question}
We address this question by providing new DIPs for planarity and related graph families. Namely, we obtain constant-round protocols with a proof size of $O(\log \log n)$ for outerplanarity, embedded planarity, series-parallel graphs, and graphs of treewidth at most $2$; and a proof size of $O(\log \log n+\log \Delta)$ for planarity in graphs of maximum degree $\Delta$ (we distinguish between embedded planarity in which we assume that a graph embedding is given in a distributed manner, and planarity in which no embedding is given; see Section \ref{section:planar} for full details). We also show that the $\Omega(\log n)$ lower bound of \cite{FeuilloleyFMRRT21} can be extended to one-round DIPs. Therefore, our results give the first evidence to the advantage provided from interaction in planarity certification.

\paragraph{Model.} In this paper, we consider \emph{distributed interactive proofs (DIPs)} based on the model of \cite{kol2018interactive}.\footnote{We note that \cite{kol2018interactive} uses the terms dAM and dMAM to denote the special cases of a DIP protocol with $2$ and $3$ rounds, respectively.} In the DIP setting, instances are graphs $G=(V,E)$ taken from some universe $\mathcal{U}$ and the goal is to distinguish between \emph{yes-instances} that come from a yes-family $\mathcal{F}_{Y}\subset \mathcal{U}$ and \emph{no-instances} that come from a no-family $\mathcal{F}_{N}=\mathcal{U}-\mathcal{F}_{Y}$.\footnote{One can easily adapt the setting so that instances also include some local information to the nodes (e.g., identifiers, weights, etc.). We chose to avoid this additional notation as the results of this paper apply to graphs without local node information.} A DIP is an interactive protocol between a distributed \emph{verifier} operating concurrently at all nodes of the graph and a centralized \emph{prover} that can see the entire instance. The prover and verifier interact back and forth in \emph{rounds}. Let $\mathcal{I}_{\mathtt{vrf}}$ denote the rounds in which the verifier interacts with the prover and let $\mathcal{I}_{\mathtt{prv}}$ denote the rounds in which the prover interacts with the verifier.

Our protocols are \emph{public-coin} which means that in each round $i\in \mathcal{I}_{\mathtt{vrf}}$, the verifier at each node $v\in V$ interacts by drawing a random bitstring $\rho_{v}^{i}\in \{0,1\}^{*}$ and sending it to the prover (in particular, the verifier cannot hide any random bits from the prover). The prover interacts with the verifier in rounds $i\in \mathcal{I}_{\mathtt{prv}}$ by sending a message $\mu_{v}^{i}\in \{0,1\}^{*}$ to each node $v\in V$. Keeping up with the terminology of \cite{KormanKP10}, we sometimes refer to the messages sent by the prover as \emph{labels}. The interaction ends with a round in which the prover interacts with the verifier, after which the verifier at each node $v\in V$ computes a local yes/no output based on: (1) the random bitstrings $\rho_{v}^{i}$ drawn by $v$ throughout the protocol; (2) the labels $\mu_{v}^{i}$ assigned to $v$ by the prover throughout the protocol; and (3) the labels $\mu_{u}^{i}$ assigned to $v$'s neighbors $u\in N(v)$ by the prover throughout the protocol. We say that the verifier \emph{accepts} the instance if all nodes output `yes', and that the verifier \emph{rejects} the instance if at least one node outputs `no'.

As standard, the correctness of a proof system is defined by \emph{completeness} and \emph{soundness} requirements. The completeness requirements asks that if $G\in \mathcal{F}_{Y}$, then there exists an \emph{honest} prover causing the verifier to accept the instance; whereas the soundness requirement asks that if $G\in \mathcal{F}_{N}$, then for any prover, the verifier rejects the instance. In the DIP setting, the correctness requirements are relaxed so that the completeness and soundness hold with probabilities $1-\epsilon_{c}$ and $1-\epsilon_{s}$, respectively, for some parameters $0\leq \epsilon_{c}, \epsilon_{s}<1/2$. In this case, we refer to $\epsilon_{c}$ as the \emph{completeness error} and to $\epsilon_{s}$ as the \emph{soundness error}. A protocol is said to have \emph{perfect completeness} if $\epsilon_{c}=0$. A DIP protocol is measured by the amount of prover-verifier communication it requires. Namely, the objective is to design protocols with a small number of interaction rounds and a small \emph{proof size} which is defined as the size of the longest label assigned by the honest prover during the protocol.

\smallskip
\noindent\textbf{The Challenge of Going Below the $\log n$ Barrier.} As observed in \cite{NaorPY20},
achieving sub-logarithmic proof lengths presents a significant challenge in the DIP setting and also serves as a lower bound for numerous problems in the non-interactive setting. This difficulty arises because many fundamental operations—such as identifying neighboring nodes, counting, or specifying node IDs — intrinsically require $\log n$ bits. While \cite{NaorPY20} made important initial progress in this area, their results apply to a more permissive variant of the DIP model, where nodes are allowed to send different messages to each of their neighbors. Indeed, their key technique is based on a rooted spanning tree provided by the prover such that every node identifies its tree-parent based on its internal port-numbering. Therefore, for each node to be able to learn its children in the tree (which is crucial to their protocols), every node has to send a distinct message to its parent.  

In contrast, our work operates within the more restrictive DIP framework defined by Kol et al. \cite{kol2018interactive}, where nodes may only forward the proofs they receive to their neighbors. This constraint aligns with the non-interactive proof labeling model introduced in \cite{KormanKP10}, where a node’s decision is based solely on its own proof and those of its neighbors. This key difference in model assumptions becomes especially important in the sub-logarithmic setting, effectively preventing us from directly applying the techniques developed in \cite{NaorPY20}. 

\paragraph{Our Results.}
We present new distributed interactive proofs for various well-studied graph families. The first graph family considered is that of path-outerplanar graphs. Previously, \cite{FeuilloleyFMRRT21} showed that path-outerplanarity admits a proof labeling scheme with a proof size of $O(\log n)$. We improve upon the communication complexity of that result by designing a protocol with exponentially shorter proof labels as specified in the following theorem.
\begin{theorem}\label{theorem:path-outerplanar}
There exists a distributed interactive proof for path-outerplanarity running in $5$ interaction rounds. The proof admits perfect completeness, a soundness error of $1/\poly\log n$, and a proof size of $O(\log \log n)$. 
\end{theorem}

Building upon the path-outerplanarity protocol, we provide a protocol for (general) outerplanarity with the same asymptotic communication guarantees.
\begin{theorem}\label{theorem:outerplanar}
		There exists a distributed interactive proof for outerplanarity running in $5$ interaction rounds. The proof admits perfect completeness, a soundness error of $1/\poly\log n$, and a proof size of $O(\log \log n)$.
\end{theorem}

We then move on to consider the case of planar graphs. In this context, we consider two verification tasks referred to as \emph{planar embedding} and \emph{planarity}. In the planar embedding task, an embedding of the graph is given in a distributed manner and the goal is to decide if it is a valid planar embedding (i.e., if no edges cross); see formal definition in Section \ref{section:planar}. In the planarity task, the goal is simply to decide if the given graph is planar. The details of our protocols for these tasks are given in the following two theorems.
\begin{theorem}\label{theorem:embedding}
		There exists a distributed interactive proof for planar embedding running in $5$ interaction rounds. The proof admits perfect completeness, a soundness error of $1/\poly\log n$, and a proof size of $O(\log \log n)$.
\end{theorem}
	
\begin{theorem}\label{theorem:planarity}
There exists a distributed interactive proof for planarity running in $5$ interaction rounds. The proof admits perfect completeness, a soundness error of $1/\poly\log n$, and a proof size of $O(\log \log n+\log \Delta)$.
\end{theorem}


We also consider the two closely related graph families of series-parallel graphs and graphs of treewidth at most $2$. We obtain the following two results.
\begin{theorem}\label{theorem:series-parallel}
There exists a distributed interactive proof for series-parallel graphs running in $5$ interaction rounds. The proof admits perfect completeness, a soundness error of $1/\poly\log n$, and a proof size of $O(\log \log n)$.
\end{theorem}

\begin{theorem}\label{theorem:treewidth}
There exists a distributed interactive proof for graphs of treewidth at most $2$ running in $5$ interaction rounds. The proof admits perfect completeness, a soundness error of $1/\poly\log n$, and a proof size of $O(\log \log n)$.
\end{theorem}
Finally, we provide the following lower bound.
\begin{theorem}\label{theorem:lb}
	For each of the following graph families, any one-round distributed interactive proof with completeness and soundness errors smaller than $1/10$ requires a proof size of $\Omega(\log n)$: (1) path-outerplanar graphs; (2) outerplanar graphs; (3) embedded planar graphs;  (4) planar graphs; (5) series-parallel graphs; and (6) graphs of treewidth at most $2$;
\end{theorem}
We note that Theorem \ref{theorem:lb} strengthens the lower bound presented in \cite{FeuilloleyFMRRT21} in the following ways. First, the lower bound of \cite{FeuilloleyFMRRT21} only applies to one-round proofs with \emph{deterministic} verifier. Theorem \ref{theorem:lb} states that the same bound holds even if the verifier is \emph{randomized}. Combined with the upper bounds stated above, our results present a strong evidence of the power added from \emph{interaction} in the context of distributed proofs for planarity and related tasks. We remark that our lower bound holds even if the nodes have access to (unbounded) \emph{shared} randomness. We also note that the lower bound of \cite{FeuilloleyFMRRT21} does not explicitly apply to some of the graph families that appear in Theorem \ref{theorem:lb} (namely, path-outerplanar graphs, embedded planar graphs, and series-parallel graphs).

\paragraph{Open problems.} 
Our results leave some intriguing unresolved questions that can be explored in follow-up works. Here, we highlight three of them. 

As the main open problem, we ask whether the additive $O(\log \Delta)$ term is necessary in the proof size for planarity. That is, we pose the following question.
\begin{oq}\label{op:1}
	Is it possible to obtain a constant round protocol for planarity with a proof size of $O(\log \log n)$ even on graphs with maximum degree $\Delta=\omega(\poly\log n)$?
\end{oq}
One may also ask whether $5$ interaction rounds are necessary in order to obtain a proof size of $O(\log \log n)$ for the tasks discussed in this paper. Of course, we know by Theorem \ref{theorem:lb} that $1$ round is insufficient. However, for any $1<r<5$, whether an $r$-round protocol exists remains open even if we simply look for a proof size of $o(\log n)$. This leads to the following open problem.
\begin{oq}\label{op:2}
	Is it possible to obtain an $r$-round protocol for e.g., outerplanarity, with a proof size of $o(\log n)$ for some $1<r<5$?
\end{oq}
Finally, we ask whether it is possible to improve our protocol's communication bound.
\begin{oq}\label{op:3}
	Is it possible to obtain a protocol for e.g., outerplanarity, where the prover communicates $o(\log \log n)$ bits with each node?
\end{oq}

\subsection{Additional Related Work} 
\paragraph{Beyond planarity.} Following the introduction of efficient distributed proof systems for planarity \cite{NaorPY20,FeuilloleyFMRRT21}, researchers have become interested in distributed proof systems for other graph families. The aforementioned compiler of \cite{NaorPY20} implies a three-round distributed interactive protocol with $O(\log n)$ proof size for families of sparse graphs (i.e., $m=O(n)$ edges) that admit a linear-time recognition algorithm. These include, e.g.,  \emph{bounded genus} graphs and \emph{outerplanar} graphs. Distributed proofs for bounded genus graphs were studied further in \cite{EsperetL22,FeuilloleyF0RRT23} where proof labeling schemes with a proof size of $O(\log n)$ are presented. For outerplanar graphs, a proof labeling scheme with a proof size of $O(\log n)$ is presented in \cite{BousquetFP24}. Additionally, the authors show similar results for a myriad of minor-free graphs.

\paragraph{Distributed interactive proofs variants.} In \cite{CrescenziFP19}, trade-offs between different parameters of the DIP model are explored. The parameters considered include the form of randomness, the complexity measures, and the number of interaction rounds. Recently, the notion of \emph{distributed quantum interactive proofs} was introduces by the authors of \cite{GallMN23} as a quantum variant of distributed interactive proofs. The main result of \cite{GallMN23} is a generic transformation from a $k$-round ``standard" proof into a $5$-round quantum proof for any constant $k>5$. Distributed quantum proofs have also been considered in a non-interactive setting in \cite{FraigniaudGNP21,HasegawaKN24}. Another exciting variant that was introduced recently in \cite{BickKO22} is that of a \emph{distributed zero-knowledge proof}. In particular, the authors adapt the classical notion of knowledge from the centralized setting (as defined in \cite{GoldwasserMR89}) to a distributed setting.
\section{Preliminaries and Definitions}\label{section:preliminaries}
\paragraph{Conventions.}
Throughout, if not specified otherwise, a graph $G=(V,E)$ is assumed to be undirected and connected. For each node $v\in V$, we stick to the convention that $N_{G}(v)$ denotes the set of $v$'s \emph{neighbors} in the graph, $E(v)$ denotes the set of edges incident on $v$, and $\deg_{G}(v)=|N_{G}(v)|=|E(v)|$ denotes $v$'s degree in $G$. Whenever $G$ is clear from the context, we may omit it from the notation and write $N(v)$ and $\deg(v)$ instead of $N_{G}(v)$ and $\deg_{G}(v)$. For a node-subset $V'\subseteq V$, we denote by $G(V')$ the subgraph induced on $G$ by $V'$. 

In the case that $G$ is directed, we assume that the edge orientation is given to the nodes such that each node $v\in V$ can distinguish between its incoming and outgoing incident edges. For a directed edge $e$ with endpoints $u$ and $v$, we write $e=(u,v)$ to reflect that $e$ is directed from $u$ to $v$, and $e=(v,u)$ otherwise. In the context of a distributed interactive proof, we assume that the label assigned by the prover to node $v\in V$ can be viewed by both its incoming and outgoing neighbors.

\paragraph{Hamiltonian paths.}
Consider a graph $G=(V,E)$ with a Hamiltonian path $P$. For a pair $u,v\in V$ of nodes, define the relation $\prec_{P}$ so that $u\prec_{P} v$ if $u$ precedes $v$ in $P$. Naturally, this extends to $u\preceq_{P} v$ if $u\prec_{P} v$ or $u=v$. Going forward, when $P$ is clear from context, we may omit it from our notation and write $u\prec v$ and $u\preceq v$ instead of $u\prec_{P} v$ and $u\preceq_{P} v$, respectively. Whenever we encounter a Hamiltonian path, it will be convenient to think of it drawn as a straight line from left to right. Keeping up with this convention, for each node $v\in V$, we can partition its non-path edges in $G$ into $v$-\emph{left} edges which are incident on neighbors $u\prec v$, and $v$-\emph{right} edges which are incident on neighbors $v\prec u$. We say that a non-path edge $(u,v)$ is the \emph{longest} $v$-left (resp., $v$-right) edge if $u\prec v$ (resp., $v\prec u$) and $u\prec u'$ (resp., $u'\prec u$) for every neighbor $u'\in N(v)$.

\paragraph{Left-right sorting.} 
We define a verification task called \emph{left-right sorting (LR-sorting)} which is used as a sub-task in our protocols. In LR-sorting, a directed graph $G=(V,E)$ is given. The graph $G$ admits a directed Hamiltonian path $P$ which is given such that each node $v\in V$ knows its incident edges in $P$. The path $P$ is assumed to be directed from left to right. The goal of the task is to decide if $u\prec v$ for every directed edge $(u,v)\in E-P$. That is, a yes-instance is defined so that $u\prec v$ for every edge $(u,v)\in E$; whereas a no-instance admits at least one edge $(u,v)\in E$ such that $v\prec u$. Observe that equivalently, yes-instances are ones in which $G$ is a DAG (in which case the $P$-ordering is the unique topological sort of $G$); and no-instances are ones in which $G$ admits some cycle.

\paragraph{Path-outerplanar graphs.}
A graph $G=(V,E)$ is said to be \emph{path-outerplanar} if it admits a Hamiltonian path $P$ such that all non-path edges can be drawn above $P$ without crossings. If the edges can be drawn in such a manner, we say that they are \emph{properly nested} within $P$ (or simply properly nested when $P$ is clear from the context). Equivalently, a graph is path-outerplanar if no two edges $(u,v),(u',v')\in E$ satisfy $u\prec_{P} u'\prec_{P} v\prec_{P} v'$ with respect to some Hamiltonian path $P$ (cf.\ \cite{FeuilloleyFMRRT21}). Refer to Figure \ref{figure:path-outerplanar} for a pictorial example of a path-outerplanar graph and some of the related definitions. 

\begin{figure}
	\centering
	\includegraphics[width=\textwidth]{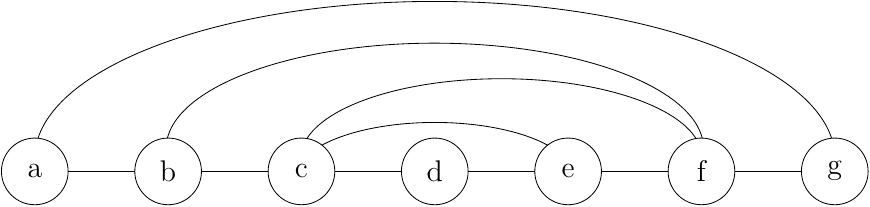}
	\caption{A path-outerplanar graph. The longest $c$-right edge is $(c,f)$; the longest $f$-left edge is $(b,f)$; the successor of $(c,e)$ is $(c,f)$.}
	\label{figure:path-outerplanar}
\end{figure}

The following simple observation will be useful in our protocol for path-outerplanar graphs in Section \ref{section:path-outerplanar}.
\begin{observation}\label{observation:path-outerplanar-longest}
	Suppose that $G$ is a path-outerplanar graph and let $(u,v)\in E$ be a non-path edge such that $u\prec v$. The edge $(u,v)$ is either the longest $u$-right edge or the longest $v$-left edge.
\end{observation}
\begin{proof}
	Assume that $(u,v)$ is neither the longest $u$-right edge nor the longest $v$-left edge. Let $(u,v')$ and $(u',v)$ be the longest $u$-right and $v$-left edges, respectively. These edges satisfy $u'\prec u\prec v\prec v'$ which contradicts path-outerplanarity.
\end{proof}

Given a path-outerplanar graph $G$, we make the following definitions. For a non-path edge $(u,v)$, $u\prec v$, define its \emph{successor} as the edge $(u',v')$ that satisfies: (1) $u'\preceq u\prec v\preceq v'$; and (2) $u''\preceq u'\prec v'\preceq v''$ for every edge $(u'',v'')$ that satisfies $u''\preceq u\prec v\preceq v''$. Intuitively, the successor of an edge is the edge drawn directly above it. For cohesiveness, for any edge $e$  that does not have a successor in the graph, define the successor to be a virtual edge $e^{*}=(u^{*},v^{*}),\ u^{*},v^{*}\notin V$, defined so that $u^{*}\prec v\prec v^{*}$ for any $v\in V$. Notice that each edge has a unique successor. Naturally, we say that $e$ is a \emph{predecessor} of $e'$ if  $e'$ is the successor of $e$. We say that two edges $e$ and $e'$ are \emph{siblings} if they have a common successor.

The following observation is now straightforward from the definitions.
\begin{observation}\label{observation:path-outerplanar-pred}
	Suppose that $G$ is a path-outerplanar graph and let $e=(u,v)$ be a (possibly virtual) non-path edge such that $u\prec v$. There exists an ordering $(u_{1},v_{1}),(u_{2},v_{2}),\dots ,(u_{k},v_{k})$ of $e$'s predecessors such that $u\preceq u_{1}\prec v_{1}\preceq u_{2}\prec v_{2}\dots \preceq u_{k}\prec v_{k}\preceq v$.
\end{observation}

\paragraph{Encoding a spanning forest in a planar graph.} As a building block in our protocol, we would like for the prover to be able to communicate a spanning forest $F$ of the graph $G$ to the verifier. While it is trivial to achieve in general using $O(\log n)$-bit labels, in our case we would like much smaller labels. It turns out that this task can be achieved in planar graphs deterministically and with constant-sized labels. This is done by slightly extending a construction of \cite{BousquetFZ24} which is designed for the task of deciding whether a planar graph admits a perfect matching.\footnote{The scheme extends to some classes of non-planar graphs; see \cite{BousquetFZ24} for full details.} We state the construction's properties in the following lemma.
\begin{lemma}\label{lemma:tree-advice}
	Let $G$ be a planar graph and let $F$ be a rooted spanning forest of $G$ (i.e., $F$ is a collection of rooted trees). For some constant $c>0$, there exists a label assignment $L:V\rightarrow\{0,1\}^{c}$ such that each node $v\in V$ can learn its parent and children in $F$ only as a function of $L(v)$, and the labels $L(u)$ assigned to $v$'s neighbors $u\in N(v)$.
\end{lemma}
For completeness of presentation, we provide a proof for the lemma. We emphasize that this construction only allows the prover to communicate the forest $F$ to the verifier and does not provide proof that $F$ is indeed a spanning forest.\footnote{Another way to formulate this construction is in terms of \emph{advice}, based on the model of \cite{FraigniaudIP10,FraigniaudKL10}. Specifically, using the terminology of \cite{FraigniaudKL10}, the statement means that computing any spanning forest of a planar graph admits an $(O(1),0)$-advising scheme.} 
\begin{proof}
	For ease of presentation, let us assume that the prover tries to communicate a spanning tree $T$ (i.e., a connected forest). The case of an unconnected forest admits a similar construction. 
	Suppose that $G=(V,E)$ is a planar graph and $T$ is a spanning tree rooted at some node $r\in V$. For each node $v\in V-\{r\}$, let $\parent(v)$ denote its parent and let $\depth(v)$ denote its depth. Define the graph $G_{odd}$ (resp., $G_{even}$) to be the graph obtained by starting from $G$ and contracting all edges $(v,\parent(v))$ that go from an odd (resp., even) depth node $v$ to its parent in $T$. Observe that $G_{odd}$ and $G_{even}$ are both planar and thus, $4$-colorable. Towards providing the label assignment, the prover computes $4$-colorings of $G_{odd}$ and $G_{even}$, respectively. For each node $v\in V$, let $c_{1}(v)$ the color of the node into which $v$ contracted in $G_{odd}$ and let $c_{2}(v)$ be the color of the node into which $v$ contracted in $G_{even}$. The prover assigns each node $v\in V$ with the label $L(v)=(c_{1}(v),c_{2}(v),\mathtt{parity}(v))$ where $\mathtt{parity}(v)=\depth (v)\bmod 2$. 
	
	We argue that the label assignment allows each node $v\in V$ to deduce which of its neighbors are its parent and children in $T$. The idea is as follows. If a node $v\in V$ of odd depth receives a color $c_{1}(v)$, then due to the validity of the coloring on $G_{odd}$, it holds that $\parent(v)$ is the only neighbor of $v$ with even depth for which $c_{1}(\parent(v))=c_{1}(v)$. The case of even depth nodes is similar. 
	
	To make things more concrete, consider some node $v\in V$ with $\mathtt{parity}(v)=1$ (resp., $\mathtt{parity}(v)=0$). Node $v$ identifies its parent as its only neighbor $u\in N(v)$ with $\mathtt{parity}(u)=0$ and $c_{1}(v)=c_{1}(u)$ (resp., $\mathtt{parity}(v)=1$ and $c_{2}(v)=c_{2}(u)$). Additionally, $v$ identifies its children as the neighbors $u\in N(v)$ that satisfy $\mathtt{parity}(u)=0$ and $c_{2}(v)=c_{2}(u)$ (resp., $\mathtt{parity}(v)=1$ and $c_{1}(v)=c_{1}(u)$). 
\end{proof}

\paragraph{Enabling edge-labels in planar graphs.}
In the technical sections, it will be convenient to describe protocols assuming that the prover can also assign edge-labels (such that both of the edge endpoints can see the label) rather than only node-labels. This assumption is facilitated by the following lemma.\footnote{Transformations that enable edge-labels in planar graphs have been presented in previous papers (see, e.g., \cite{FeuilloleyFMRRT21}). However, these constructions require the prover to assign an ordering to the nodes which incurs an additive $\Theta (\log n)$ overhead to the label size. Thus, we cannot use these transformations for our purposes.}
\begin{lemma}\label{lemma:edge-labels}
	Let $\Pi$ be a class of planar graphs. Suppose that there exists a distributed interactive proof deciding whether $G\in\Pi$ in which the prover assigns labels of size $\ell$ to the nodes and edges. Then, there exists a distributed interactive proof in which the prover assigns labels of size $O(\ell)$ only to the nodes. Furthermore, the two proofs admit the same number of interaction rounds.
\end{lemma}
\begin{proof}
	It is well-known that planar graphs have arboricity at most $3$. This means that the edge-set of any planar graph $G=(V,E)$ can be partitioned into three edge-disjoint forests $F_{1},F_{2},F_{3}$. By Lemma \ref{lemma:tree-advice}, the prover can inform each node $v\in V$ of its parent and children in each $F_{i}$ using only constant-sized labels. Then, instead of assigning a label $L(u_{i},v)$ to the edge $(u_{i},v)$ between $v$ and its parent $u_{i}$ in $F_{i}$, the prover simply writes $L(u_{i},v)$ to a field in $v$'s label which is designated for its parent in $F_{i}$. This allows both endpoints to learn the label $L(u_{i},v)$, thus enabling the simulation of edge-labels.
\end{proof}

\paragraph{Spanning tree verification.}
Consider a graph $G=(V,E)$ and let $T$ be a subgraph of $G$ such that each node $v\in V$ knows its incident edges in $T$. We define \emph{spanning tree verification} as the task of deciding whether $T$ is a spanning tree of $G$. The following lemma is established in \cite{NaorPY20}.
\begin{lemma}[\protect{\cite[Section 7.1]{NaorPY20}}]\label{lemma:npy-spanning-tree}
	There is a distributed interactive proof for spanning tree verification with $3$ interaction rounds and constant proof size. The proof admits perfect completeness and a constant soundness error.
\end{lemma}
\noindent Observe that by standard parallel repetition, one can reduce the soundness error to $1/2^\ell$ at the expense of a $\Theta(\ell)$ proof size for any parameter $\ell>0$. Throughout the paper, this fact will be used in a black-box manner. 

\paragraph{Multiset equality.} In the \emph{multiset equality} problem, each node $v\in V$ receives as input two multisets $S_{1}(v),S_{2}(v)$ and the goal is to decide whether $S_{1}=S_{2}$, where $S_{1}$ and $S_{2}$ are the multisets $S_{1}=\bigcup_{v\in V}S_{1}(v)$ and $S_{2}=\bigcup_{v\in V}S_{2}(v)$. Notice that in the  definition of $S_{1}$ and $S_{2}$, the union is taken with respect to multisets, i.e., the multiplicity of element $s$ in $S_{1}$ (resp., $S_{2}$) is the sum of its multiplicities over all multisets $S_{1}(v)$ (resp., $S_{2}(v)$).  The multisets $S_{1}$ and $S_{2}$ are assumed to be of size at most $k$ for some integer $k>0$, and the elements are taken from a universe of size $k^{c}$ for some constant $c\geq 1$. For our purposes, it would also be convenient to assume that the nodes are given a distributed encoding of a rooted spanning tree of the graph. The following lemma can be derived from the multiset equality protocol of \cite{NaorPY20}.
\begin{lemma}[\cite{NaorPY20}]\label{lemma:multiset-equality-npy}
	Given a multiset equality instance $(G,S_{1},S_{2},k)$ such that $|S_{1}|,|S_{2}|\leq k$ and a rooted spanning tree $T$ of $G$, there exists a $2$-round distributed interactive proof for multiset equality. The proof admits perfect completeness, a soundness error of $1/k^{c}$, and a proof size of $O(\log k)$.
\end{lemma}
Since the lemma above is not explicit in \cite{NaorPY20} and since we use the details of the multiset equality protocol in a white-box manner, we describe here the construction's details. The multiset equality protocol relies on the following idea. For a multiset $S$, define the polynomial $\varphi_{S}(x)=\prod_{s\in S}(s-x)$.\footnote{Notice that we assume here w.l.o.g.\ that multiset elements are integers.} Now, observe that $S_{1}=S_{2}$ if and only if $\varphi_{S_{1}}\equiv \varphi_{S_{2}}$. Moreover, notice that the degree of $\varphi_{S_{1}}(x)$ and $\varphi_{S_{2}}(x)$ is at most $k$. Define $p$ as the smallest prime number that satisfies $p>k^{c+1}$ and let $z$ be a variable drawn uniformly at random from $\{0,\dots ,p-1\}$. By polynomial identity testing properties, if $\varphi_{S_{1}}\not\equiv \varphi_{S_{2}}$ and $\varphi_{S_{1}}(z),\varphi_{S_{2}}(z)$ are computed over the field $\mathbb{F}_{p}$, then  $\Pr[\varphi_{S_{1}}(z)=\varphi_{S_{2}}(z)]\leq k/p\leq 1/k^{c}$.\footnote{Recall that $\mathbb{F}_{p}$ is the field whose elements are $\{0,\dots, p-1\}$ and operations are done modulo $p$.}

Following this idea, the multiset equality problem essentially reduces to evaluating a polynomial at a random point $z\in \{0,\dots ,p-1\}$. Recalling that we assume that a distributed encoding of a rooted spanning tree is provided to the nodes, the polynomial evaluation is implemented as follows. First, the point $z\in \{0,\dots ,p-1\}$ is sampled by the root and sent to the prover. Then, the prover assigns each node $v\in V$ with the value $z$ and the values $\varphi_{S^{v}_{1}}(z),\varphi_{S^{v}_{2}}(z)$ (computed over $\mathbb{F}_{p}$) where $S_{1}^{v}$ (resp., $S_{2}^{v}$) is the multiset of elements from $S_{1}$ (resp., $S_{2}$) in $v$'s subtree. Given the assigned values, it is well-known that their validity can be checked at each node $v\in V$ based on its input, its label, and its children's labels. This is because polynomial evaluation is an aggregation task that can be verified ``up the tree'' (see, e.g., \cite[Lemma 4.4]{KormanKP10} for details). Following these checks, the root $r$ can check that $\varphi_{S^{r}_{1}}(z)=\varphi_{S^{r}_{2}}(z)$ (which implies that $\varphi_{S_{1}}(z)=\varphi_{S_{2}}(z)$).

To summarize, given a rooted spanning tree of the graph, the multiset equality protocol runs for $2$ interaction rounds, admits perfect completeness, a soundness error of $1/k^{c}$, and a proof size of $O(\log p)= O(\log k)$.\footnote{Recall that standard density of primes properties assure that $p$ cannot be too large. In particular, $p<k^{c+2}$, and thus $\log p=O(\log k)$.}
\section{Technical Overview}\label{section:overview}
In this section, we provide an overview of the techniques used to obtain the protocols presented in the paper. To exemplify the challenge of planarity certification with labels of size $O(\log \log n)$, let us first sketch a seemingly natural (yet unsuccessful) approach to which we refer as the \emph{clustering} approach. Suppose that the prover computes a partition of the graph into node-disjoint connected clusters of size $\poly\log n$. Then, the prover provides a proof that: (1) the subgraph induced by each cluster is planar; and (2) the graph obtained by contracting all clusters is planar. Notice that in terms of \emph{proof size} this approach is promising (and indeed, a similar approach was used in \cite{NaorPY20} to achieve sub-logarithmic proofs for various other problems). This is because one can use , e.g., the logarithmic proof for planarity of \cite{FeuilloleyFMRRT21} on each cluster to obtain a proof of size $O(\log \log n)$ for (1). As for (2), since each cluster has size $\poly \log n$ and acts as a single node in the contracted graph, one can hope that it is possible to distribute the proof of a single cluster among the nodes within that cluster using only $O(\log \log n)$-sized labels per node. For the sake of this example, let us assume that it is indeed possible to obtain a proof for (2) with a proof size of $O(\log \log n)$. 

It is not hard to see that the proof provided from the clustering approach is complete --- if $G$ is planar, then so are the subgraphs induced by the clusters as well as the graph obtained from contracting every cluster. The fundamental problem however, is that a proof for planarity obtained from this approach \emph{cannot} be sound. To see why this is the case, consider a non-planar graph that contains a single $5$-clique $H$.\footnote{Recall that a graph is planar if and only if it does not contain $K_{5}$ or $K_{3,3}$ as a minor.} A cheating prover can then define the partition such that, e.g., two nodes of $H$ are assigned to one cluster and the other three are assigned to a different cluster. In this case, $H$ does not violate planarity within any cluster and translates to a single edge in the contracted graph. Therefore, in this instance the verifier is likely to accept which violates the soundness requirement. 

Notice that even if somehow we were able to prevent the prover from constructing an adversarial partition, it is possible to construct a no-instance in which the verifier is likely to accept for \emph{any} partition. For example, it is possible to create an instance where each edge of $H$ is subdivided such that its endpoints are at distance $\Omega(n)$ in $G$ (and thus, are separated in any partition). Hence, we conclude that the clustering approach is doomed to fail for the planarity task. 

The inherent failure of the clustering approach compels us to come up with a different approach. Similarly to the approach of \cite{FeuilloleyFMRRT21}, we seek to reduce planarity tasks to some other well-structured task for which we are able to design an efficient protocol. Perhaps surprisingly, we show that efficient protocols for all tasks considered in the current paper can be obtained based on a protocol for the seemingly unrelated task of LR-sorting. Indeed, starting from LR-sorting we present a sequence of reductions leading to new protocols for outerplanarity, planar embedding, planarity, series-parallel, and graphs of treewidth $\leq 2$. Refer to Figure \ref{fig:summary} for a chart depicting the dependencies between our different constructions. Notice that an advantage of LR-sorting is that unlike the planarity task, it can benefit from the clustering approach. As we explain below in more detail, this becomes useful in our protocol.

\begin{figure}[t!]
	\centering

	\begin{tikzpicture}
	[
	every node/.style={
		align = center,
		fill          =  black!5,
		inner sep     = 6pt,
		minimum width = 22pt,
		rounded corners=.1cm}
	]
	
			\node (sorting) at (-14,5) {LR-Sorting \\ (Lemma \ref{lemma:sorting})};
		\node (pathouterplanar) at (-14,2) {Path-Outerplanarity \\ (Thm. \ref{theorem:path-outerplanar})};
	\node (outerplanar) at (-19.5,2) {Outerplanarity \\ (Thm. \ref{theorem:outerplanar})};
	\node (embplanar) at (-8.3,2) {Embedded \\Planarity   \\ (Thm. \ref{theorem:embedding})};
	\node (planar) at (-3.5,2) {Planarity \\ (Thm. \ref{theorem:planarity})};
	\node (seriesparallel) at  (-14,-0.3){Series-Parallel \\ (Thm. \ref{theorem:series-parallel})};
	\node (treewidth) at (-14,-2.7) {Treewidth $\leq 2$ \\ (Thm. \ref{theorem:treewidth})};
	
		\draw[->, line width=1.4, black] (sorting) -- (pathouterplanar) node[midway,left,fill=transparent!0]{\textcolor{black}{Lemma \ref{lemma:sorting-to-path-outerplanarity}}};
	
	\draw[->, line width=1.4, black] (pathouterplanar) -- (outerplanar) node[midway,below,fill=transparent!0]{\textcolor{black}{}};
	
	\draw[->, line width=1.4, black] (pathouterplanar) -- (embplanar) node[midway,below,fill=transparent!0]{\textcolor{black}{Lemma \ref{lemma:path-outerplanarity-to-embedding}}};
	
		\draw[->, line width=1.4, black] (embplanar) -- (planar) node[midway,below,fill=transparent!0]{\textcolor{black}{Lemma \ref{lemma:embedding-to-planarity}}};
	
		\draw[->, line width=1.4, black] (pathouterplanar) -- 
	(seriesparallel) node[midway,left,fill=transparent!0]{\textcolor{black}{}}; 
	
			\draw[->, line width=1.4, black] (seriesparallel) -- 
	(treewidth) node[midway,left,fill=transparent!0]{\textcolor{black}{}}; 
	
	\end{tikzpicture}

	\caption{High-level description of the main results and their connections. In the case where there is no reference next to an arrow $x \to y$, the protocol for $y$ is obtained by using the protocol for problem $x$ in a white-box manner. 
	}
	\label{fig:summary}
\end{figure}
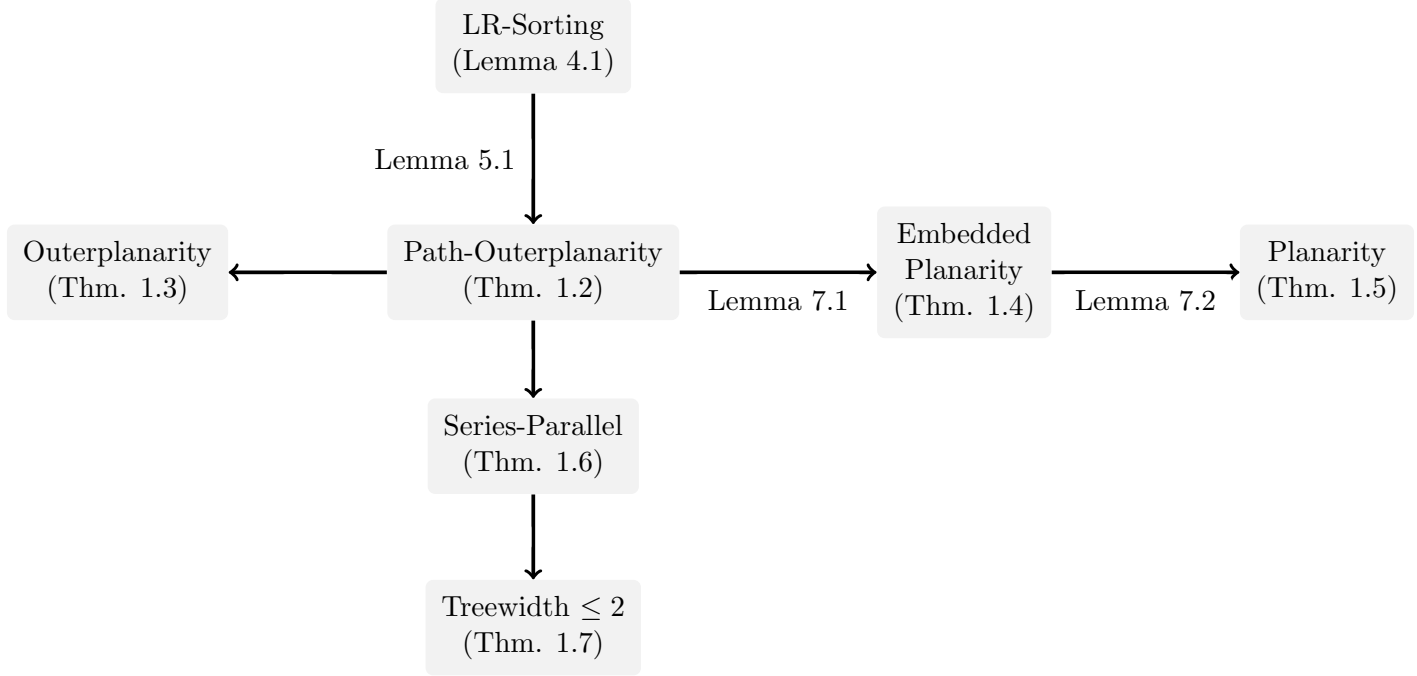

\paragraph{LR-sorting.}  
To give an intuition for the LR-sorting protocol, let us first sketch a simple one-round protocol for LR-sorting with a proof size of $O(\log n)$. The prover assigns each node with its position on the path. Then, each node $v\in V$ that is located at the $i$-th position can verify that: (1) its path-neighbors are located at positions $i\pm 1$; and (2) all of its outgoing edges go towards nodes of larger positions. 

To obtain a distributed interactive proof with a proof size of $O(\log \log n)$, the idea is to divide the nodes into node-disjoint \emph{blocks} where each block is made up of  $\lceil \log n\rceil $ consecutive path-nodes. This allows the prover to distribute the position of block $b$, denoted by $\pos(b)$, such that each node receives: (1) its \emph{index} $i\in [\lceil \log n\rceil]$ within $b$; and (2) $\pos(b)[i]$, i.e., the $i$-th most significant bit of $b$'s position. Ideally, as in the clustering approach, we would like for each block to act as a single node in the trivial protocol. In this short description, let us assume for simplicity that the prover encodes the position of the blocks correctly according to the $P$-ordering and let $\pos(b)$ denote the position of block $b$. The main challenge of the protocol is then captured by the following question: suppose that there is an edge $(u,v)$ where $u$ and $v$ belong to different blocks $b_{u},b_{v}$, how can the prover prove to $u$ and $v$ that $\pos(b_{u})<\pos(b_{v})$?

Towards answering this question, let us first consider the simpler case where $(u,v)$ is the \emph{only} edge leaving the block for both $b_{u}$ and $b_{v}$. Furthermore, we will describe the protocol under the assumption that the prover can assign edge-labels. Recall that Lemma \ref{lemma:edge-labels} implies that the edge-labels assumption can be simulated in planar graphs while incurring only a constant overhead to the proof size. Let $i$ be the index of the most significant bit in which $\pos(b_{u})$ and $\pos(b_{v})$ differ. Notice that by definition, if $\pos(b_{u})<\pos(b_{v})$, then $\pos(b_{u})[i]=0$ and $\pos(b_{v})[i]=1$. In addition, $i$ is in the range $[\lceil \log n\rceil ]$ and thus, can be encoded using $O(\log \log n)$ bits. In the first interaction, the prover encodes $i$ to the label of $(u,v)$. The prover then needs to prove that (1) $\pos(b_{u})[i]=0$; (2) $\pos(b_{v})[i]=1$; and (3) $\pos(b_{u})$ and $\pos(b_{v})$ agree on their $i-1$ most significant bits. 

Let us focus on how (3) is proved. A straightforward way to obtain a proof for (3) is to assign the edge $(u,v)$ with the substring containing the $i-1$ most significant bits in their blocks. The main problem with this approach is that the proof size could be as large as $\Theta(\log n)$. To avoid this large label size, we can use polynomial identity testing. That is, we interpret the substrings containing the $i-1$ most significant bits in $\pos(b_{u})$ and $\pos(b_{v})$ as polynomials of degree $O(\log n)$ and seek to have the prover and verifier evaluate them at a random point over a field of size $\poly \log n$. The prover then assigns the outcome of the polynomial evaluation to the label of $(u,v)$. This introduces the following \emph{locality} problem: the verification that the polynomials were computed correctly is done at the $(i-1)$-th leftmost node in each block (using the standard aggregation technique of \cite{KormanKP10}). Since $u$ and $v$ might be far from the $(i-1)$-th leftmost node in their respective blocks, they cannot locally detect that the prover is not lying about the outcome. We note that in the simplified case where $(u,v)$ is the only edge leaving $b_{u}$ and $b_{v}$, the problem is easy to solve. Indeed, since $(u,v)$ is the only edge considered by blocks $b_{u}$ and $b_{v}$, the prover can assign the outcome of the polynomial evaluation at the $(i-1)$-th node to \emph{all} nodes of $b_{u}$ and $b_{v}$. The nodes can then check the correctness of this assignment locally based on the assignment to their block-neighbors. 

Moving on to the general case where each block may have \emph{many} edges leaving it, it is not clear how to solve the locality problem described above. Of course, if the prover assigns each node of the block with the outcomes of polynomial evaluations for every relevant index, this could require assigning $\Theta(\log n)$ values to every node which completely defeats the purpose. The main observation that we make is that one can formulate the locality problem as an instance of \emph{multiset equality} between two multisets that are carefully defined within each block. Then, to check whether the multisets are indeed equal, a multiset equality protocol is executed within the blocks. An important property of the constructed multisets is that they are of size $\poly \log n$ which means that this final step can be done while maintaining a proof size of $O(\log \log n)$. 

For the correctness, it turns out that the protocol's completeness becomes straightforward from the multisets construction, whereas the soundness argument requires a bit more care. Essentially, we show that in a no-instance the prover is likely to \emph{commit} to a pair of unequal multisets within some block. Then, conditioning on this event, it is likely that the verifier rejects the instance due to the soundness of the multiset equality protocol.

\paragraph{From LR-sorting to path-outerplanarity.} 
In Section \ref{section:path-outerplanar}, we devise a protocol for path-outerplanarity. To get intuition of how the protocol works, we first recall the labels assigned in the non-interactive proof of \cite{FeuilloleyFMRRT21}. Essentially, the prover assigns each node $v\in V$ with its position in the path $P$ along with the position of the nodes $u$ and $u'$ for which $(u,u')\in E$ is the first edge that is drawn above $v$. The authors of \cite{FeuilloleyFMRRT21} show that this labeling allows the (deterministic) verifier to verify that indeed $G$ is path-outerplanar.

Of course, assigning the positions of nodes in $P$ is far too costly for our purposes, so we seek to avoid it by using interaction and randomization. To that end, we first note that the assignment of positions in \cite{FeuilloleyFMRRT21} serves the following purposes: (1) each node learns its $P$-neighbors; (2) each edge can be identified based on the positions of its endpoints; and (3) each node learns a clockwise orientation of its incident edges. Note that since $P$ is a spanning tree, (1) can be obtained using only constant-sized labels based on Lemmas \ref{lemma:tree-advice} and \ref{lemma:npy-spanning-tree}. For (2), we show that the edge identifiers can be replaced by random bits. As for (3), we design a protocol in which it is sufficient for every node to only distinguish between its left and right edges with respect to the path $P$. In fact, we show that under the assumption that every node knows its left and right edges, path-outerplanarity can be solved in $3$ rounds and with a \emph{constant} proof size  (refer to Lemma \ref{lemma:sorting-to-path-outerplanarity} for the full details of the reduction). To lift this assumption, we can apply our LR-sorting protocol which leads to the stated complexity bound of Theorem \ref{theorem:path-outerplanar}.

\paragraph{From path-outerplanarity to outerplanarity and planarity.} We handle outerplanarity and planarity through reductions to path-outerplanarity. We note that while such reductions are presented in previous works (\cite{FeuilloleyFMRRT21} for planarity; \cite{BousquetFP24} for outerplanarity), we cannot use them as-is in our setting. This is because these reductions incur an additive $\Theta (\log n)$ overhead to the proof size. Nevertheless, our results rely on some modifications of the existing reductions. For outerplanarity, our reduction is white-box and it avoids the $\Theta (\log n)$ overhead based on the tools of Lemma \ref{lemma:tree-advice} and Lemma \ref{lemma:npy-spanning-tree} as well as some observations regarding the path-outerplanarity protocol. 

For planarity, we start from the planar embedding task as an intermediate point. To reduce planar embedding to path-outerplanarity, we revisit some of the constructive proofs presented in \cite{FeuilloleyFMRRT21} and show that they can be used to obtain such a reduction. Once again, Lemma \ref{lemma:tree-advice} and Lemma \ref{lemma:npy-spanning-tree} are used for the sake of efficient implementation. Then, we show that planarity reduces to planar embedding while incurring only an additive $O(\log \Delta)$ overhead to the proof size, thus obtaining the stated result.

\paragraph{Lower bounds.}
The starting point for our lower bounds is the lower bound of \cite{FeuilloleyFMRRT21} which applies to proof labeling schemes (in fact, it holds also for the more general locally checkable proofs \cite{GoosS16}). The result of \cite{FeuilloleyFMRRT21} shows that an $\Omega(\log n)$ proof size is required even for the task of deciding whether a graph is outerplanar or non-planar. We start by adjusting the lower bound's details so that it would apply for the task of deciding whether a graph is \emph{biconnected} outerplanar or non-planar.\footnote{Recall that a graph is biconnected if the removal of any node leaves the resulting graph connected. A biconnected component of a graph $G$ is a maximal biconnected subgraph of $G$.} This adjustment leads to a bound for all the graph families considered in the current paper since they are planar and contain all biconnected outerplanar graphs. To extend the lower bound to one-round protocols in which the verifier is randomized, we use a framework presented in \cite{FraigniaudMORT19}.


\section{LR-Sorting Protocol}\label{section:sorting}
In this section, we present a protocol for the task of LR-sorting on a given directed graph $G=(V,E)$ with Hamiltonian path $P$. The protocol is implemented under the assumption that the prover is able to assign labels to the nodes and the edges of $G$. If a label $L(u,v)$ is assigned to the edge $(u,v)\in  E$, then both endpoints $u$ and $v$ can view it. The main result of the current section is the following lemma.
\begin{lemma}\label{lemma:sorting}
	There exists a distributed interactive proof for LR-sorting running in $5$ interaction rounds. The proof admits perfect completeness, a soundness error of $1/\poly\log n$, and a proof size of $O(\log \log n)$ where labels are assigned to both nodes and edges.
\end{lemma}
Recall that by Lemma \ref{lemma:edge-labels}, if the given graph is planar, we can lift the edge-labels assumption of Lemma \ref{lemma:sorting} to get the following.
\begin{lemma}\label{lemma:sorting-planar}
	There exists a distributed interactive proof for LR-sorting in planar graphs running in $5$ interaction rounds. The proof admits perfect completeness, a soundness error of $1/\poly\log n$, and a proof size of $O(\log \log n)$.
\end{lemma}
The rest of the section is dedicated to the protocol's description. Recall that our goal is to show how the prover proves that $u\prec v$ for every non-path edge $(u,v)$ directed from $u$ to $v$. This is described in two stages. First, a division of the path into node-disjoint \emph{blocks} is described. Then, we explain how the block construction allows the nodes to compare their relative position on the path. For clarity, the stages are described without regard for the number of interaction rounds. Then, by the end of the section, we explain how the protocol can be implemented in $5$ interaction rounds. Throughout, $c>0$ is defined as a positive constant that can be made large enough to support the protocol's soundness guarantee.

\subsection{The Block Construction}\label{section:sorting-block}
The block construction is defined so that the first block consists of the  $\lceil \log n\rceil$ leftmost nodes in the path, the second block consists of the next $\lceil \log n\rceil$ nodes and so on. For ease of presentation, we assume that all blocks are of size exactly $\lceil \log n\rceil$. One can easily adjust the protocol's details to handle the general case in which (only) the rightmost block may have more than $\lceil \log n\rceil$ (but less than $2\lceil \log n\rceil$) nodes.  

The purpose of the block construction is to allow the nodes to receive information  regarding their position on the path. The position of a block $b$, denoted by $\pos(b)$, is defined to be $i-1$ if $b$ is the $i$-th leftmost block. Notice that due to the block size, it is possible to encode an integer $x\in\{0,\dots ,n-1\}$ through the nodes of a block using only $O(\log \log n)$ bits per block. To do so, assign the $j$-th leftmost node of the block with the number $j$ as well as the $j$-th most significant bit of $x$ (leading zeros are added if necessary). Using this mechanism, the prover assigns the values $\pos(b)$ and $\pos(b)+1$ to each block $b$. 

In addition, the prover provides a proof that the two numbers assigned to each block are consecutive. To explain how this is done, suppose that $x$ is a nonnegative integer in binary representation and let $j$ be its least significant bit valued $0$. Notice that $x$ and $x+1$ differ (only) in their $j$ least significant bits. For a block $b$, define $j_{b}$ to be the least significant bit in $\pos(b)$ whose value is $0$ and let $v_{b}$ be the node associated with the index $j_{b}$. To prove that the numbers assigned to $b$ are consecutive, the prover marks $v_{b}$ and informs every other node in the block whether it is to the right/left of $v_{b}$. 

To present the verification process at block $b$, let us denote by $x_{1}(b)$ and $x_{2}(b)$ the bitstrings assigned to $b$ under the claim $x_{1}(b)+1=x_{2}(b)$. If a node $v\in b$ was labeled to be to the right of $v_{b}$, then it checks that its bit in $x_{1}(b)$ is $1$, its bit in $x_{2}(b)$ is $0$, and its right neighbor in the block (if such neighbor exists) is also labeled as right of $v_{b}$; if $v$ was marked as $v_{b}$, then it checks that its bit in $x_{1}(b)$ is $1$, its bit in $x_{2}(b)$ is $0$, its right neighbor is labeled as right of $v_{b}$, and its left neighbor is labeled as left of $v_{b}$; and if $v$ is labeled as left of $v_{b}$, then it checks that it received the same bit in both bitstrings and that its left neighbor is labeled as left of $v_{b}$.

To complete the block construction stage, the verifier checks that the position assignment is consistent between adjacent blocks.  Let $b$ and $b'$ be two adjacent blocks where $b'$ is to the right of $b$. The verifier seeks to check that $\pos(b)+1=\pos(b')$. To that end, the multiset equality protocol is used between $x_{2}(b)$ and $x_{1}(b')$, where a bitstring is interpreted as the subset of $[\lceil \log n\rceil]$ that contains the indices whose bit is $1$. Notice that the sets (and so, the degree of the multiset equality polynomials) are of size at most $\lceil \log n\rceil$. The verifier and prover run the multiset equality protocol over the field $\mathbb{F}_{p}$ where $p$ is the smallest prime satisfying $p> \log ^{c}n$.  The polynomials are computed at a random point $r\in\{0,\dots ,p-1\}$ which is the same for all blocks. To that end, the variable $r$ is sampled by the leftmost node in the path and passed to all nodes in the graph by the prover. Each block $b$ computes (with the prover's assistance) the values of the two polynomials associated with its encoded bitstrings $x_{1}(b),x_{2}(b)$. This allows every pair of adjacent blocks to check that the positions assigned to them are indeed consecutive. 

\paragraph{Correctness.} If the position assignment given by the prover is valid, then $x_{2}(b)=x_{1}(b')$ for every pair of adjacent blocks $b,b'$ and by the completeness of the multiset equality protocol, the verifier does not reject in this case. On the other hand, if the position assignment is not valid, then at least one pair $b,b'$ of adjacent blocks satisfies $x_{2}(b)\neq x_{1}(b')$ and by the soundness of the multiset equality protocol, the verifier rejects with probability $1-\lceil \log n\rceil /p= 1-1/\poly\log n$.

\paragraph{Remark.} An alternative approach to verifying the validity of the block construction is to use the RAM compiler of \cite{NaorPY20} concurrently on pairs of consecutive blocks. Nevertheless, the approach and notations presented above will be useful in the presentation of the next stage. We also note that it might be plausible to implement the block construction stage with proof size of $o(\log\log n)$, we avoided these optimizations as the key ``communication bottleneck" lies in the next stage.

\subsection{Comparing Relative Positions}\label{section:sorting-comparisons}
We now describe how the prover uses the block construction to prove claims of the form $u\prec v$ for all non-path edges $(u,v)$. To that end, we divide the edges into two types as follows. The \emph{inner-block} edges are defined as the edges $(u,v)$ in which $u$ and $v$ belong to the same block, and the \emph{outer-block} edges are defined as the edges $(u,v)$ in which $u$ and $v$ belong to different blocks.

\paragraph{Inner-block edges.} Suppose that $(u,v)$ is an inner-block edge. To show that $u\prec v$, the prover first assigns a bit to the edge $(u,v)$ indicating that it is an inner-block edge. Let us denote the indices of $u$ and $v$ within their block by $i_{u}$ and $i_{v}$, respectively (recall that these indices were assigned to the nodes during the block construction stage). The nodes $u$ and $v$ check that $i_{u}<i_{v}$ and if not, reject immediately. If $i_{u}<i_{v}$, then it is left to check that $u$ and $v$ are indeed on the same block. To that end, the leftmost node of each block $b$ (i.e., the node associated with the most significant bit of $\pos(b)$) samples a number $r_{b}\in[\log ^{c}n]$ and sends it to the prover which in response, sends the value $r_{b}$ to all nodes of the block $b$. Each node checks that the number it received is consistent with its block neighbors and the leftmost node in the block checks that it received the same number it sampled. Then, for every edge $(u,v)$ that was labeled as inner-block, $u$ and $v$ check that they both received the same $r_{b}$ value and reject otherwise. 

\paragraph{Correctness for inner-block edges.} For completeness, observe that if $u\prec v$, then all checks succeed and the verifier accepts. For soundness, if $v\prec u$ and $u$ and $v$ are on the same block, then it must hold that $i_{v}<i_{u}$ and the verifier rejects. So, suppose that $v$ and $u$ are on different blocks $b\neq b'$ but the prover labels $(u,v)$ as an inner-block edge. Then, the verifier rejects unless $r_{b}=r_{b'}$ which happens with probability $1/\poly\log n$.

\paragraph{Outer-block edges.} We complete the protocol's description by addressing the case of outer-block edges. Consider an outer-block edge $(u,v)$, i.e., $u$ and $v$ belong to different blocks $b_{u}$ and $b_{v}$, respectively. The prover's goal is to show that $\pos(b_{u})<\pos(b_{v})$. We divide the proof into two parts referred to as the \emph{commitment} scheme and the \emph{verification} scheme. 

The main idea behind the commitment scheme relies on the following simple fact. Suppose that $x$ and $y$ are two nonnegative integers represented by binary strings of same length (leading zeros are added if necessary). Then, $x<y$ if and only if there exists an index $i$ such that the $i-1$ most significant bits of $x$ and $y$ are identical, the $i$-th most significant bit of $x$ is $0$, and the $i$-th most significant bit of $y$ is $1$. We shall refer to this index as the \emph{$(x,y)$-distinguishing} index and denote it by $I(x,y)$. 

Consider a non-path edge $(u,v)$ whose endpoints belong to different blocks $b_{u}$ and $b_{v}$, respectively. The commitment scheme starts by having the prover write the value $I_{u,v}=I(\pos (b_{u}),\pos (b_{v}))$ to the label of edge $(u,v)$. Then, for each block $b$, the multiset equality polynomial that is associated with $\pos(b)$ is computed at a random point $r'\in \{0,\dots, p-1\}$, where $p$ is the prime number defined above. Similarly to the block construction stage, the computation is done over the finite field $\mathbb{F}_{p}$ and the variable $r'$ is the same for all blocks. For an index $i\in[\lceil \log n\rceil]$, let $\pos(b)[1,\dots ,i]$ denote the substring of $\pos(b)$ consisting of its $i$ most significant bits and let us denote by $\varphi^{b}_{i}$ the multiset equality polynomial that is identified with the substring $\pos(b)[1,\dots ,i]$. We note that $\varphi^{b}_{i}(r')$ is exactly the value computed at the $i$-th leftmost bit of block $b$. In addition to computing the multiset equality polynomial values within the blocks, the prover writes the value $\varphi^{b_{u}}_{I_{u,v}-1}(r')$ (which is equal to $\varphi^{b_{v}}_{I_{u,v}-1}(r')$ by the definition of the distinguishing index) to the label of each non-path edge $(u,v)$. 

For an edge $e=(u,v)$ which is classified as an outer-block edge, let $\rho(e)=(i,j)$ be the pair of values assigned by the prover on the label of $e$ during the commitment scheme. That is, here $i$ is claimed by the prover to be the distinguishing index between $u$ and $v$'s block positions, and $j$ is claimed to be the multiset equality polynomial value computed at the $(i-1)$-th index of both blocks. To complete the commitment scheme, the verifier at each node $v\in V$ makes some consistency checks. First, if the same index $i$ appears as the first element in two pairs $\rho(u,v)$ and $\rho(v,u')$ associated with edges $(u,v),(v,u')$, then the verifier rejects. To see why this condition is imposed, notice that $u\prec v$ requires that the $i$-th bit of $v$'s block is $1$, whereas $v\prec u'$ requires that the $i$-th bit of $v$'s block is $0$. For the second consistency check, the verifier checks that if two of $v$'s incident edges agree on the first element of $\rho(\cdot)$ (and did not fail the first check), then they agree on the second element of $\rho(\cdot)$. For each node $v\in V$, let us define $C_{0}(v)$ (resp., $C_{1}(v)$) as the set of pairs $\rho(u,v)$ (resp., $\rho(v,u)$) assigned by the prover to the edges $(u,v)$ (resp., $(v,u)$) during the commitment scheme. Notice that we define $C_{0}(v),C_{1}(v)$ as sets and not multisets. In particular, this means that $|C_{0}(v)|+|C_{1}(v)|\leq \lceil \log n\rceil$.

The purpose of the verification scheme is to verify the validity of the values in $C_{0}(v)$ and $C_{1}(v)$ for each node $v\in V$. Notice that this cannot be achieved locally in a trivial manner since the indices that appear in $C_{0}(v)$ and $C_{1}(v)$ may be associated with nodes on $v$'s block that are not adjacent to $v$. We describe the verification of $C_{1}(v)$ values and then explain the small change required for the $C_{0}(v)$ verification. For a block $b$, define $C_{1}(b)$ as the \emph{multiset} $C_{1}(b)=\bigcup_{v\in b}C_{1}(v)$. Define $F(b)$ as the set of indices whose bit in $\pos (b)$ is equal to $1$ and let $D_{1}(b)=\bigcup_{i\in F(b)}\{(i,\varphi^{b}_{i-1}(r'))\}$. The main idea behind the verification scheme is that in yes-instances, for each node $v\in b$ and pair $(i,j)\in C_{1}(v)$, it holds by construction that $(i,j)\in D_{1}(b)$. Thus, one can construct a multiset which is equal to $C_{1}(b)$ by taking every element of $D_{1}(b)$ with some multiplicity between $0$ and $\lceil \log n\rceil$ (notice that it is not guaranteed that every element of $D_{1}(b)$ is in $C_{1}(b)$, so we allow a ``multiplicity" of $0$). 
Following this idea, the validity of $C_{1}(b)$ can be verified by means of another multiset equality protocol.


To make things more concrete, consider a node $v\in b$ which is associated with an index $i_{v}\in F(b)$. The prover provides $v$ with a value $M_{v}\in\{0,\dots ,\lceil\log n\rceil\}$ that counts the number of times the pair $(i_{v},\varphi^{b}_{i_{v}-1}(r'))$ appears in $C_{1}(b)$. Then, the prover and verifier execute a multiset equality protocol to compare between $C_{1}(b)$ and the multiset obtained by taking $M_{v}$ copies of the pair $(i_{v},\varphi^{b}_{i_{v}-1}(r'))$ for each $i_{v}\in F(b)$. Here, notice that node $v$ gets the value $i_{v}$ from its own label and the value $\varphi^{b}_{i_{v}-1}(r')$ from the label of its left neighbor on the path. When computing the multiset equality polynomials, each pair $(i,j)\in [\lceil\log n\rceil]\times\{0,\dots ,p-1\}$ is mapped to an element from the set $[p\cdot \lceil\log n\rceil]$ by means of a fixed bijection known in advance to all nodes. To accommodate this range of field elements, it suffices to execute the multiset equality protocol over the field $\mathbb{F}_{p'}$ such that $p'$ is the smallest prime that satisfies $p'>p\cdot \lceil\log n\rceil$. To verify the validity of $C_{0}(b)=\bigcup_{v\in b}C_{0}(v)$, we apply a similar idea with respect to the set $D_{0}(b)=\bigcup_{i\notin F(b)}\{(i,\varphi^{b}_{i-1}(r'))\}$. Observe that all the multisets that are involved in the equality protocols (and thus, the degrees of all polynomials) are of size $O(\log ^{2}n)$.

\paragraph{Correctness for outer-block edges.} The completeness follows directly from the definition of the distinguishing index and the completeness of the multiset equality protocol. Regarding soundness, suppose that for some edge $(u,v)$ directed from $u$ to $v$, it holds that $v\prec u$. Denote by $(i,j)$ the pair assigned to the edge $(u,v)$ by the prover in the commitment scheme.

First, consider the case that $u$ and $v$ are in the same block $b$ (but $(u,v)$ is labeled as an outer-block edge by the prover). Notice that $(i,j)$ can be in at most one of the sets $D_{0}(b),D_{1}(b)$. This is because $i\in F(b)$ implies $(i,j)\notin D_{0}(b)$ and $i\notin F(b)$ implies $(i,j)\notin D_{1}(b)$. Assume w.l.o.g.\ that $(i,j)\notin D_{0}(b)$. Notice that by construction $(i,j)\in C_{0}(b)$, which means that the compared multisets cannot be equal. Hence, by the soundness of the multiset equality protocol, the verifier rejects in this case with probability $1-\log ^{2}n/(p\log n)=1-1/\poly \log n$.  

Now, suppose that $u$ and $v$ are in different blocks $b_{u}\neq b_{v}$. Notice that by construction, $(i,j)\in C_{0}(b_{u})$ and $(i,j)\in C_{1}(b_{v})$. If $i\in F(b_{u})$ or $i\notin F(b_{v})$, then the soundness follows from a similar argument to the former case. Otherwise, by the definition of the distinguishing index and by the soundness of the multiset equality protocol, it follows that $\varphi^{b_{u}}_{i-1}(r')\neq \varphi^{b_{v}}_{i-1}(r')$ with probability $1-1/\poly\log n$. If this is the case, then it must be that either $j\neq \varphi^{b_{u}}_{i-1}(r')$ or $j\neq \varphi^{b_{v}}_{i-1}(r')$. Let us condition on this event and assume w.l.o.g.\ that $j\neq \varphi^{b_{u}}_{i-1}(r')$. Then, $(i,j)\notin D_{0}(b_{u})$ and since $(i,j)\in C_{0}(b_{u})$, the soundness of the multiset equality protocol suggests that the verifier rejects with probability $1-1/\poly\log n$. 

\subsection{Protocol's Complexity}\label{section:sorting-communication}
For ease of presentation, our protocol is described in separate stages. Here, we observe that parts of the stages can be parallelized. First, we observe that the block construction stage can be implemented in three interaction rounds. Indeed, it starts with  the prover encoding the block positions along with a proof that each block receives two consecutive numbers. Then, the verifier interacts with the prover to compute two multiset equality polynomials within each block. This can be done in two additional interaction rounds for a total of three. Similarly, the proofs of $u\prec v$ for inner-block edges $(u,v)$, and the commitment scheme of outer-block edges can be completed within three rounds. Moreover, a correct execution of these steps does not depend on the execution of the block construction, thus they can be executed in parallel. We also note that the multiplicity values $M_{v}$ that are presented in the verification stage of outer-block edges can actually be precomputed by the prover and assigned during the first interaction (they are placed in the verification scheme strictly for the sake of clear presentation). Therefore, after three interaction rounds, it is the verifier's turn to speak and the remaining task is the multiset equality protocol of the verification scheme of outer-block edges (here, notice that the verification scheme cannot be executed sooner as it depends on the values assigned in the commitment scheme). This takes two additional interaction rounds for a total of five rounds. Regarding proof size, a bound of $O(\log \log n)$ is straightforward from the construction.

  		\section{Path-outerplanarity}\label{section:path-outerplanar}
	 In this section, we present a protocol that uses LR-sorting as a sub-task to decide whether a given graph is path-outerplanar. The properties of the protocol are specified in the following lemma.
	 \begin{lemma}\label{lemma:sorting-to-path-outerplanarity}
	 	Suppose that there exists a distributed interactive proof for LR-sorting verification in planar graphs running in $t$ interaction rounds. Let $\ell$ be the proof size, $\epsilon_{c}$ be the completeness error, and $\epsilon_{s}$ be the soundness error of the LR-sorting protocol. Then, there is a distributed interactive proof for path-outerplanarity running in $\max\{t,3\}$ rounds and admitting a proof size of $O(\ell)$, a completeness error of $\epsilon_{c}$, and a soundness error of $\epsilon_{s}+2^{-\ell}$.   
	 \end{lemma}
As a consequence, we get Theorem \ref{theorem:path-outerplanar}.
 	\begin{proof}[Proof of Theorem \ref{theorem:path-outerplanar}]
 		The protocol is obtained by plugging the LR-sorting protocol of Lemma \ref{lemma:sorting-planar} into the statement of Lemma \ref{lemma:sorting-to-path-outerplanarity}. 
 	\end{proof}
	The rest of the section is dedicated to the description of the protocol that proves Lemma \ref{lemma:sorting-to-path-outerplanarity}. For clarity, the protocol is described in separate stages without regard for the number of interaction rounds. Then, by the end of the section, we explain how the protocol can be implemented within the desired amount of interaction. Throughout, $c>0$ is defined as a positive constant that can be made large enough to support the protocol's soundness guarantee.
	
	\paragraph{Committing to a path.} 
	The protocol starts by having the prover commit to a Hamiltonian path $P$ of $G$. To encode $P$, the prover uses the labels of Lemma \ref{lemma:tree-advice} where $P$ is rooted at the leftmost node in the path. Each node can verify that it has at most one child in the given tree encoding. Additionally, to verify that the given subgraph is indeed a Hamiltonian path of the graph, the prover and verifier execute the protocol of Lemma \ref{lemma:npy-spanning-tree} amplified by means of a $c\cdot \ell$ parallel repetition. 
	
	Observe that if the graph is indeed path-outerplanar, then the prover can successfully send the verifier a Hamiltonian path. Consequently, each node knows its path-edges and is able to differentiate between its right and left neighbor on the path. On the other hand, if the graph is not path-outerplanar, then the graph is either not Hamiltonian or not outerplanar. In the former case, the prover is not able to provide a Hamiltonian path which causes the verifier to reject with probability $1-2^{-\Theta(\ell)}$; in the latter case, the prover is able to send the verifier a Hamiltonian path but the non-path edges are not properly nested.
	
	\paragraph{LR-sorting.}  This stage starts by having the prover inform the verifier whether $u\prec v$ or $v\prec u$ for every edge $(u,v)\in E$. To see how this is achieved, recall that in the simulation of edge-labels that proves Lemma \ref{lemma:edge-labels}, $e$'s label is written within the label of one of its endpoints. Let us refer to that endpoint as the endpoint \emph{accountable} for $e$. So, if $u$ is accountable for $e$, then the prover assigns the bit $1$ to $u$'s sub-label associated with $e$ to signify that $u\prec v$, and $0$ otherwise. 
	
	Following this assignment, the goal of the verifier is to check that all edges were labeled correctly by the prover, i.e., to check that if an edge $(u,v)$ was labeled $u\prec v$, then indeed $u$ appears before $v$ in $P$. To that end, the prover and verifier execute an LR-sorting protocol. To create an instance for LR-sorting, the edges of the graph are oriented according to the prover's labeling. That is, if edge $(u,v)$ was labeled $u\prec v$, then it is oriented from $u$ to $v$. 
	
	Notice that if the verifier accepts the LR-sorting instance, then this means that the prover labeled all edges correctly (up to a soundness error of $\epsilon_{s}$). So, for the rest of the protocol, we assume that for every non-path edge $e=(u,v)$, both endpoints know whether $u\prec v$ or $v\prec u$. Notice that in particular, this means that each $v\in V$ can distinguish between its left and right edges.
	
	\paragraph{Nesting verification.}  
	In this final stage, the goal is to verify that the non-path edges are properly nested. The stage starts with the prover informing the endpoints of each non-path edge $e=(u,v)$, $u\prec v$, whether it is the longest $u$-right edge and whether it is the longest $v$-left edge. This is done by assigning two bits within the label of the endpoint accountable for $e$ similarly to the previous stage.
	
	
	
	Upon receiving the edge-labels, the verifier at each node $v\in V$ runs the following checks. If $v$ has any right (resp., left) edges, then the verifier checks that exactly one of them is marked as longest $v$-right (resp., $v$-left) edge. In addition, for every right (resp., left) edge $(v,u)$ that was not marked longest $v$-right (resp., $v$-left), the verifier checks that it was marked longest $u$-left (resp., $u$-right). If one of the checks fail, then the verifier immediately rejects. Otherwise, $v$ samples a bitstring $s_{v}\in \{0,1\}^{c\cdot \ell}$ uniformly at random and sends it to the prover. For each non-path edge $(u,v)$ such that $u\prec v$, define its \emph{name} to be the pair $(s_{u},s_{v})$. 
	
	After receiving the $s_{v}$ values from all nodes, the prover assigns to each edge $e$ its name through a sub-label $\mathtt{name}(e)$ and its successor's name through a sub-label $\mathtt{succ}(e)$ where the name of the virtual edge $e^{*}=(u^{*},v^{*})$ is defined by the designated symbol $\bot$ (recall that $e^{*}$ is the successor of edges with no real successor in the graph). Additionally, if edge $e=(u,v)$ has predecessors $(u_{1},v_{1}),\dots ,(u_{k},v_{k})$ such that $u\preceq u_{1}\prec v_{1}\preceq \dots \preceq u_{k}\prec v_{k}\preceq v$, then the prover assigns the label $\mathtt{above}(w)=\mathtt{name}(e)=(s_{u},s_{v})$ to every node $w$ such that $(u\prec w\preceq u_{1})\lor (v_{1}\preceq w\preceq u_{2})\lor\dots\lor  (v_{k}\preceq w\prec  v)$. In other words, the prover assigns $e$'s name to all nodes for which $e$ is the first edge drawn entirely above them (including the endpoints of $e$'s predecessors; excluding the endpoints of $e$). In particular, if $e=(u,v)$ does not have any predecessors, then $\mathtt{above}(w)=\mathtt{name}(e)$ for all nodes $w$ such that $u\prec w\prec v$. Observe that by definition, each node is associated with only one such edge and thus, receives only one edge name. 
	
	Consider a label assignment to the nodes and non-path edges. First, for each non-path edge $e$, its endpoints verify that $\mathtt{name}(e)$ is consistent with their sampled values. Then, each node $v\in V$ checks that there exists an ordering $e_{1}^{+},\dots ,e_{k}^{+}$ of its right edges, and an ordering $e_{1}^{-},\dots ,e_{k'}^{-}$ of its left edges such that the following conditions are satisfied:
	\begin{enumerate}
		\item $e_{k}^{+}$ and $e_{k'}^{-}$ are marked as the longest $v$-right and $v$-left edges, respectively.
		\item $\mathtt{succ}(e_{i}^{+})=\mathtt{name}(e_{i+1}^{+})$ for all $1\leq i<k$, and $\mathtt{succ}(e_{i}^{-})=\mathtt{name}(e_{i+1}^{-})$ for all $1\leq i<k'$.
		\item $\mathtt{above}(v)=\mathtt{succ}(e_{k}^{+})=\mathtt{succ}(e_{k'}^{-})$.
		\item if $u$ is $v$'s right neighbor on the path, then $\mathtt{name}(e_{1}^{+})=\mathtt{above}(u)$ if the set of $v$'s right edges is non-empty, and $\mathtt{above}(v)=\mathtt{above}(u)$ otherwise.
		\item if $u$ is $v$'s left neighbor on the path, then $\mathtt{name}(e_{1}^{-})=\mathtt{above}(u)$ if the set of $v$'s left edges is non-empty, and $\mathtt{above}(v)=\mathtt{above}(u)$ otherwise.
	\end{enumerate}
	We note that a pair of orderings that satisfies the described conditions does not have to be unique. Also, notice that nodes which are not incident on any non-path edges only need to check that they were assigned the same value as their neighbors on the path (conditions (4) and (5)). This concludes the description of the nesting verification. We go on to establish its correctness.
	
	\paragraph{Correctness of nesting verification.} 
	Towards proving the completeness and soundness of the nesting verification, we show the following two observations. 
	\begin{observation}\label{observation:path-outerplanar-longest-labeling}
		Fix some node $u\in V$. If the prover marks the longest $u$-right or the longest $u$-left edge incorrectly, then the verifier rejects the instance with probability $1-2^{-c\cdot \ell}$.
	\end{observation}
	\begin{proof}
		Suppose that edge $(u,v)$ is the longest $u$-right edge but not marked as such. Recall that by the initial verification conditions, $(u,v)$ must be marked as the longest $v$-left edge (otherwise the verifier rejects). If $v$ has a right edge, then by verification conditions (1) and (3), the value $\mathtt{succ}(u,v)$ should be identical to the value $\mathtt{succ}(v,w)$, where $(v,w)$ is the right edge of $v$ which is marked as longest. If $v$ does not have a right edge, then by verification conditions (3) and (4), the value $\mathtt{succ}(u,v)$ should be identical to the value the value $\mathtt{above}(w')$ where $w'$ is $v$'s right neighbor on the path. In either case, following the verification conditions we get that $\mathtt{succ}(u,v)$ should be identical to $\mathtt{name}(u',v')$ of some edge $(u',v')$ such that $v\prec v'$. Moreover, the edge $(u',v')$ is fully determined by the marking of longest left and right edges by the prover (and in particular, determined before the sampling of names). Note that since $(u,v)$ is the longest $u$-right edge and $v\prec v'$, it must hold that $(u,v')\notin E$ and thus, $u'\neq u$. On the other hand, since $(u,v)$ is not marked as the longest $u$-right edge, by condition (2), the first element of $\mathtt{succ}(u,v)$ should be $s_{u}$. So, the verifier rejects unless $s_{u}=s_{u'}$ which happens with probability $2^{-c\cdot \ell}$. The case of longest left edges follows a similar reasoning.
	\end{proof}
	Going forward with the correctness proof, we shall assume that all longest left/right edges are marked correctly. For two nodes $u\prec v$, denote by $P_{u,v}$ the set of nodes on the $(u,v)$-subpath in $G$.
	\begin{observation}\label{observation:path-outerplanar-succ-labeling}
		 Suppose that for a non-path edge $(u,v)$, it holds that $G(P_{u,v})$ is path-outerplanar w.r.t.\ $P_{u,v}$ (i.e., the edges of $G(P_{u,v})$ are properly nested within $P_{u,v}$). If the verifier accepts the instance, then $\mathtt{succ}(u',v')$ is the name of the successor of $(u',v')$ in $G(P_{u,v})$ for all non-path edges $(u',v')\neq(u,v),\ u\preceq u'\prec v'\preceq v$.
	\end{observation}
	\begin{proof}
		Let $(x,y)$ be a non-path edge in $G(P_{u,v})$ and let $(x_{\ell},y_{\ell})$ and $(x_{r},y_{r})$ be its leftmost and rightmost predecessors, respectively. First, if $y_{r}=y$, then it must hold that $x\prec x_{r}$ which means that $(x_{r},y_{r})=(x_{r},y)$ is not the longest $y$-left edge. Therefore, by condition (2) it must hold that the second element of $\mathtt{succ}(x_{r},y_{r})$ is $s_{y}$. Now, suppose that $y_{r}\prec y$. Here, since $(x_{r},y_{r})$ is a predecessor of $(x,y)$, it follows that $(x_{r},y_{r})$ is the longest $y_{r}$-left edge. Applying conditions (1), (3), and (4) along the $(y_{r},y)$-path, we once again get that the second element of $\mathtt{succ}(x_{r},y_{r})$ must be $s_{y}$. For similar reasoning, we can deduce that the first element of $\mathtt{succ}(x_{\ell},y_{\ell})$ is $s_{x}$. Now, we observe that by conditions (1), (3), (4), and (5), every pair of adjacent siblings must have the same $\mathtt{succ}(\cdot)$ field. Therefore, every predecessor $(x',y')$ of $(x,y)$ must satisfy $\mathtt{succ}(x',y')=(s_{x},s_{y})=\mathtt{name}(x,y)$ which concludes our proof.
	\end{proof}
	We can now prove the completeness and soundness of our protocol.
	\begin{lemma}\label{lemma:path-outerplanar-nesting}
		The described nesting verification admits perfect completeness and a soundness error of $2^{-\Theta(\ell)}$.
	\end{lemma}
	\begin{proof}
		We start from completeness. First, we note that by Observation \ref{observation:path-outerplanar-longest}, the honest prover can mark each edge $(u,v)$ as longest $u$-right/$v$-left correctly. Furthermore, observe that the feasibility of the $\mathtt{name}(\cdot)$, $\mathtt{succ}(\cdot)$, and $\mathtt{above}(\cdot)$ labels assigned by the honest prover is guaranteed by Observation \ref{observation:path-outerplanar-pred}. Now, consider some node $v\in V$ and let $e_{k'}^{-}=(v,u_{k'}^{-}),\dots, e_{1}^{-}=(v,u_{1}^{-}),e_{1}^{+}=(v,u_{1}^{+}),\dots ,e_{k}^{+}=(v,u_{k}^{+})$ be its incident non-path edges such that $u_{k'}^{-}\prec \dots \prec u_{1}^{-}\prec v\prec u_{1}^{+}\prec \dots \prec u_{k}^{+}$. Given the labels assigned by the honest prover, the orderings $e_{1}^{-},\dots, e_{k'}^{-}$ and $e_{1}^{+},\dots ,e_{k}^{+}$ defined on the left and right edges of $v$ satisfy all the verification conditions, thus causing the verifier to accept.
		
		We now establish the soundness guarantee. Let us define $(u,v)$ as an edge that admits a crossing edge $(u',v')$ such that  $u\prec u'\prec v\prec v'$ but not a crossing edge $(u',v')$ such that $u'\prec u\prec v'\prec v$. That is, $(u,v)$ is not crossed by edges that has an endpoint to the left of $u$. Moreover, assume that $(u,v)$ is the deepest nested such edge, i.e., every edge $(x,y)\neq (u,v)$ where  $u\preceq x\prec y\preceq v$  does not admit a crossing edge. Observe that if there exists a pair of crossing edges in the graph, then there exists an edge $(u,v)$ satisfying the assumptions stated above.
		
		We start from the case where $(u,v)$ is the longest $v$-left edge. Define $u\prec u'\prec v$ to be the rightmost node incident on a right edge that crosses $(u,v)$ and define $(u',v')$ as the longest $u'$-right edge (by definition, $(u',v')$ crosses $(u,v)$). Let $e_{1}=(u_{1},v),e_{2}=(u_{2},v),\dots, e_{k}=(u_{k},v)$ be $v$'s left edges ordered such that  $u=u_{k}\prec \dots \prec u_{2}\prec u_{1}\prec v$. Observe that by the assumptions on $(u,v)$, it follows that $u'\preceq u_{k-1}$. Moreover, all edges that are drawn below $e_{k-1}$ are properly nested. Therefore, Observation \ref{observation:path-outerplanar-succ-labeling} implies that if the verifier accepts the instance, then $\mathtt{succ}(e_{i})=\mathtt{name}(e_{i+1})$ for every $1\leq i<k-1$. Furthermore, recall that $(u,v)$ is marked as the longest $v$-left edge. Thus, for condition (2) to be satisfied at node $v$, it must also hold that $\mathtt{succ}(e_{k-1})=\mathtt{name}(e_{k})=(s_{u},s_{v})$. 
		
		To show that the verifier is likely to reject in this case, the idea is to define a sequence of edges that must agree with $e_{k-1}$ on their $\mathtt{succ}(\cdot)$ value, but also must have $s_{u'}$ as their $\mathtt{succ}(\cdot)$ value's first element. This implies that the verifier rejects unless $s_{u}=s_{u'}$ which happens with probability $1/\poly\log n$. The sequence $(x_{1},y_{1}),\dots ,(x_{t},y_{t})$ of edges is defined as follows. Start by taking $x_{1}\preceq u_{k-1}$ to be the closest node to $u_{k-1}$ incident on a left edge and set $(x_{1},y_{1})$ as the longest $x_{1}$-left edge. Then, take $x_{2}\preceq y_{1}$ to be the closest node to $y_{1}$ incident on a left edge and set $(x_{2},y_{2})$ as the longest $x_{2}$-left edge. Continue this process until reaching $y_{t}$ such that all nodes $w$ such that $u'\prec w\preceq y_{t}$ are not incident on a left edge. Notice that the sequence construction is feasible since by our assumption on $(u,v)$, no edge within $(u,v)$ crosses $(u',v')$ (and thus, the sequence is entirely to the right of $u'$). Moreover, since $u'$ is the rightmost node incident on a right edge crossing $(u,v)$, it follows that every edge $(x_{i},y_{i})$ in the sequence is the longest $y_{i}$-right edge. We note that the verification conditions dictate that every pair of adjacent edges in the sequence should have the same $\mathtt{succ}(\cdot)$ value and that $\mathtt{succ}(x_{1},y_{1})=\mathtt{succ}(e_{k-1})$. On the other hand, for $u'$ to satisfy condition (4), the first element of $\mathtt{succ}(x_{t},y_{t})$ must be $s_{u'}$ which concludes the soundness for this case. 
		
		We move on to the case where $(u,v)$ is not the longest $v$-left edge. If $(u,v)$ is also not the longest $u$-right edge, then by the construction the verifier rejects. So, assume that $(u,v)$ is the longest $u$-right edge. Let $u\prec u'\prec v$ be the leftmost node incident on an edge crossing $(u,v)$ and let $(u',v')$ be the longest $u'$-right edge (by definition, $(u',v')$ crosses $(u,v)$). By similar measures to the previous case, it is possible to find a sequence $(x_{1},y_{1}),\dots ,(x_{t},y_{t})$ of edges such that must satisfy $\mathtt{succ}(x_{1},y_{1})=\dots= \mathtt{succ}(x_{t},y_{t})=\mathtt{succ} (u',v')$; and the first element of $\mathtt{succ}(x_{i},y_{i})$ is $s_{u}$ for each $1\leq i\leq t$. On the other hand, by similar reasoning to the one presented in the proof of Observation \ref{observation:path-outerplanar-longest-labeling}, it must hold that $\mathtt{succ} (u',v')=\mathtt{name}(w,w')$ for some edge $(w,w')$ such that $v'\preceq w'$. Moreover, this edge is fully determined by the marking of longest left and right edges by the prover (and in particular before the sampling of names). Recall that $v\prec v'\preceq w'$ and that $(u,v)$ is the longest $u$-right edge. Therefore, it must hold that $w\neq u$ which means that the probability of $s_{u}=s_{w}$ is at most $2^{-c\cdot \ell}$. 
	\end{proof}
	\paragraph{Analysis of the protocol.}
	By construction, the proof size of the protocol is $O(\ell)$. Moreover, all stages apart from the black-box use of the LR-sorting protocol admit perfect completeness and a soundness error of $2^{-\Theta(\ell)}$. Thus, by union bound arguments, the completeness error of the protocol is $\epsilon_{c}$ and the soundness error is $\epsilon_{s}+2^{-\Theta(\ell)}$. Of course, taking a sufficiently large $c$, we can have a soundness error of $\epsilon_{s}+2^{-\ell}$ as desired. Finally, regarding the number of interaction rounds, we note that all stages can be executed in parallel without affecting the correctness of the algorithm. It is straightforward to see that the stages committing to a path and nesting can be implemented in $3$ interaction rounds. Since the LR-sorting protocol requires $t$ rounds, we get that in total the protocol runs in $\max\{t,3\}$ rounds.  

	\section{Outerplanarity}\label{section:outerplanar}
	In this section, we extend the protocol of Theorem \ref{theorem:path-outerplanar} from path-outerplanar graphs to (general) outerplanar graphs. Particularly, we show the following theorem.
	To design the protocol of Theorem \ref{theorem:outerplanar}, we adapt the approach of \cite{BousquetFP24} which 
	 (i) shows that path-outerplanar graphs and biconnected outerplanar graphs are \emph{almost} the same; and (ii) uses a decomposition of the graph into its biconnected components. The following Theorem will be useful as part of the protocol of Theorem  \ref{theorem:outerplanar}.
	\begin{theorem}\label{theorem:biconnected-outerplanar}
		There exists a distributed interactive proof deciding if a graph is a biconnected outerplanar graph running in $5$ interaction rounds. The proof admits perfect completeness, a soundness error of $1/\poly\log n$, and a proof size of $O(\log \log n)$.
	\end{theorem}
	\begin{proof}
		A biconnected outerplanar graph can be drawn on the plane as a Hamiltonian cycle with all non-cycle edges drawn inside it without crossings. As observed in \cite{BousquetFP24}, this implies that a biconnected outerplanar graph is path-outerplanar with respect to a Hamiltonian path $P$ such that the endpoints of $P$ are connected by an edge. Therefore, we obtain a protocol for biconnected outerplanar graphs from the protocol for path-outerplanarity in Theorem \ref{theorem:path-outerplanar} simply by adding a verification condition that there is an edge between the endpoints of $P$.
	\end{proof}
	Towards proving Theorem \ref{theorem:outerplanar}, we make the following definitions. For a graph $G=(V,E)$, a node $v$ is referred to as a \emph{cut node} if it belongs to more than one biconnected component. The \emph{block-cut} tree of $G$ is defined to be a tree $T$ in which each node is associated with either a cut node in $G$ or a biconnected component of $G$. The edges of $T$ are defined so that each cut node $v\in V$ is connected to all biconnected components $C\subseteq V$ for which $v\in C$. Suppose that the block-cut tree $T$ is rooted at some biconnected component $R\subseteq V$. For each biconnected component $C\neq R$, we refer to the cut node which is the parent of $C$ in $T$ as \emph{$C$-separating}. We are now prepared to prove Theorem \ref{theorem:outerplanar}.
	\begin{proof}[Proof of Theorem \ref{theorem:outerplanar}]
		The prover computes the block-cut tree $T$ of the graph and roots it at some biconnected component $R\subseteq V$.  For each biconnected component $C\neq R$, define $P_{C}$ to be a Hamiltonian path of $G(C)$ that emerges from the $C$-separating node. Define the \emph{$C$-leader} to be the node that neighbors the $C$-separating node in $P_{C}$, let $e_{C}$ denote the edge between the $C$-separating node and the $C$-leader, and let $P'_{C}=P_{C}-\{e_{C}\}$ be the subpath of $P_{C}$ that starts from the $C$-leader. For the root component $R$, define $P_{R}$ to be some Hamiltonian path of $G(R)$, the $R$-leader as the leftmost node of $P_{R}$, and $P'_{R}=P_{R}$. We describe the protocol in three stages.
		
		The purpose of the first stage is to verify that every node which is not a cut node is adjacent only to nodes in its component. To that end, the prover assigns each node with two bits indicating if it is a cut node and if it is a $C$-leader for some component $C$. Additionally, the prover encodes the subpaths $P'_{C}$ and edges $e_{C}$ associated with biconnected components $C$ by means of Lemma \ref{lemma:tree-advice}. In response, the verifier at each node $v\in V$ that was marked as either cut node or leader draws a random bitstring $s_{v}$ of length $\Theta (\log\log n)$ and sends it to the prover. Then, the prover sends each node $v\in P'_{C}$ the values $\mathtt{sep}(v)$ and $\mathtt{lead}(v)$ which are the random bitstrings drawn by the $C$-separating node and the $C$-leader, respectively. Each node $v\in P'_{C}$ checks that its path neighbors received the same $\mathtt{sep}(\cdot)$ and $\mathtt{lead}(\cdot)$ values. In addition, if $v$ is not a cut node, than it checks that for every neighbor $u\in N(v)$, either $\mathtt{sep}(v)=\mathtt{sep}(u)$ and $\mathtt{lead}(v)=\mathtt{lead}(u)$; or $u$ is a cut node and $\mathtt{sep}(v)=s_{u}$. Finally, the $C$-leader $v$ checks that $\mathtt{sep}(v)=s(u)$ where $u$ is its neighbor on $e_{C}$.
		
		In the second stage, the nodes verify the tree structure of $T$. To that end, it suffices to check that the subgraph $F$ which is obtained by taking the union of paths $P_{C}$ over all biconnected components $C$, is a spanning tree of $G$. This verification is done by means of the protocol in Lemma \ref{lemma:npy-spanning-tree} amplified by means of a $\Theta (\log\log n)$-repetition.
		
		Finally, the prover needs to prove that each subgraph $G(C)$ induced by a biconnected component $C$ is an outerplanar graph. To that end, we would like to use the protocol of Theorem \ref{theorem:biconnected-outerplanar} on all subgraphs $G(C)$ in parallel. The obstacle here is that cut nodes may belong to many biconnected components. Thus, a naive implementation of these parallel executions may result in a large overhead to the label size of cut nodes. 
		
		We overcome the obstacle as follows. First, for each component $C$, the prover assigns each node $v\in C$, the value $d(C)$ defined as the distance from $C$ to $R$ in $T$ modulo $3$. Notice that each $C$-separating node receives two values --- $d(C)$ and $(d(C)-1)\bmod 3$. Checking the correctness of this assignment can be done by standard measures (see, e.g., \cite{BousquetFP24,KormanKP10}). Notice that the $C$ separating node is the only node in $C$ that was assigned the values $d(C)$ and $d(C)-1\bmod 3$. Therefore, each node $v\in C$ can determine which of its neighbors is the $C$-separating node. Let us denote by $v_{\mathtt{sep}}(C)$ the $C$-separating node. The the path-outerplanarity protocol on $G(C)$ is implemented as follows. The randomness of $v_{\mathtt{sep}}(C)$ is drawn by the $C$-leader passed to the rest of the nodes in $P'_{C}$ through the prover.  In addition, the labels that are meant to be assigned to $v_{\mathtt{sep}}(C)$ are deferred to all of its neighbors. This allows each neighbor $u$ of $v_{\mathtt{sep}}(C)$ such that $u\in C$ to simulate the nesting verification for its incident edges. We note that in this case, the verification of $v_{\mathtt{sep}}(C)\prec u$ is not necessary since by definition, $v_{\mathtt{sep}}(C)$ is the leftmost node of $P_{C}$. To summarize, this implementation allows the nodes to verify that $G(C)$ is a biconnected outerplanar graph for each component $C$ while maintaining the $O(\log \log n)$ proof size. 
		
		Overall, the three stages can run in parallel for a total of $5$ interaction rounds and each stage requires a proof size of $O(\log \log n)$, has perfect completeness, and a $1/\poly\log n$ soundness error.
		\end{proof}
		 
	\section{Planar Embedding and Planarity}\label{section:planar}
	In this section, we consider the \emph{planar embedding} verification problem on a graph $G=(V,E)$. In this problem, a drawing of the graph is given to the nodes in a distributed manner by assigning each node with a \emph{clockwise ordering} of its incident edges. More formally, for each node $v\in V$, a clockwise ordering of its incident edges $E(v)$ is given in the form of a bijection $\rho_{v}:E(v)\to \{0,\dots ,\deg(v)-1\}$ that maps each edge $e\in E(v)$ to a value $\rho_{v}(e)\in  \{0,\dots ,\deg(v)-1\}$. An edge $e'\in E(v)$ comes immediately after $e\in E(v)$ in the clockwise ordering if $\rho _{v}(e')=(\rho_{v}(e)+1)\bmod \deg(v)$. Let us denote $\rho(G)=\{\rho_{v}\mid v\in V\}$. The goal in this problem is to decide if $\rho(G)$ induces a combinatorial planar embedding of $G$, i.e., if $G$ can be drawn in accordance with the clockwise orderings in $\rho(G)$ such that no two edges cross. The main technical objective of the current section is to prove the following lemma which depicts a connection between path-outerplanarity and planar embedding.
	\begin{lemma}\label{lemma:path-outerplanarity-to-embedding}
		Suppose that there exists a distributed interactive proof for path-outerplanarity running in $t$ interaction rounds. Let $\ell$ be the proof size, $\epsilon_{c}$ be the completeness error, and $\epsilon_{s}$ be the soundness error of the path-outerplanarity protocol. Then, there is a distributed interactive proof for planar embedding running in $\max\{t,3\}$ rounds and admitting a proof size of $O(\ell)$, a completeness error of $\epsilon_{c}$, and a soundness error of $\epsilon_{s}+2^{-\ell}$.   
	\end{lemma}
	 Before proving Lemma \ref{lemma:path-outerplanarity-to-embedding}, we note that it leads to the Theorem \ref{theorem:embedding}.
	\begin{proof}
		The protocol is obtained by plugging the path-outerplanarity protocol of Theorem \ref{theorem:path-outerplanar} into the statement of Lemma \ref{lemma:path-outerplanarity-to-embedding}. 
	\end{proof}
	We also consider the \emph{planarity} problem in which the goal is to decide whether a given graph $G$ is planar. A reduction between the problems is given in the following lemma.
	\begin{lemma}\label{lemma:embedding-to-planarity}
		Suppose that there exists a distributed interactive proof for planar embedding running in $t$ interaction rounds. Let $\ell$ be the proof size, $\epsilon_{c}$ be the completeness error, and $\epsilon_{s}$ be the soundness error of the planar embedding protocol. Then, there is a distributed interactive proof for planarity running in $t$ rounds and admitting a proof size of $\ell+O(\log \Delta)$, a completeness error of $\epsilon_{c}$, and a soundness error of $\epsilon_{s}$.   
	\end{lemma}
	\begin{proof}
		Given a planar graph $G=(V,E)$, the prover first computes a combinatorial planar embedding of $G$. Let $\rho(G)=\{\rho_{v}\mid v\in V\}$ be a collection of bijections $\rho_{v}:E(v)\to \{0,\dots ,\deg(v)-1\}$ that encode the clockwise orderings of the computed embedding as described above. The idea is to have the prover send each node $v\in V$ the values $\rho_{v}(e)$ of its incident edges and then use the protocol of Theorem \ref{theorem:embedding} to prove that they induce a valid embedding. So, it remains to explain how the prover can pass $\rho_{v}$ to $v$ using $O(\log \Delta)$ bits.
		
		Recall that by Lemma \ref{lemma:edge-labels}, the prover and verifier can simulate edge-labels in the graph $G$ incurring only a constant overhead to the proof size. Moreover, as part of the construction, the prover encodes a decomposition of the edges into three rooted forests (such that each node knows its parent in each forest). Based on that, the prover can provide each node $v$ with the values $\rho_{v}(e)$ of its incident edges $e\in E(v)$ as follows. Consider an edge $e=(u,v)\in E$ and assume w.l.o.g.\ that $u$ is $v$'s parent in the forest decomposition of $G$. The prover writes the (ordered) pair $(\rho_{u}(e),\rho_{v}(e))$ to $e$'s label. Notice that this encoding is achieved using only $O(\log \Delta)$ bits for each node and that consequently, each node can learn the values $\rho_{v}(e)$ of all its incident edges $e\in E(v)$. 
		
		The correctness relies on the fact that $G$ is planar if and only if it admits a combinatorial planar embedding. Therefore, if $G$ is planar, then the prover can provide clockwise orderings that correspond to a combinatorial planar embedding of $G$. The completeness now follows from the completeness of the protocol stated in Theorem \ref{theorem:embedding}. Regarding soundness, if $G$ is not planar, then no valid combinatorial planar embedding of $G$ exists. Thus, any clockwise orderings assigned to the nodes do not induce a combinatorial planar embedding. The soundness now follows from the soundness of the protocol stated in Theorem \ref{theorem:embedding}.
	\end{proof}
	\noindent Theorem \ref{theorem:planarity} follows by Lemma \ref{lemma:embedding-to-planarity} and Theorem \ref{theorem:embedding}.
	We move on to describe the protocol of Lemma \ref{lemma:path-outerplanarity-to-embedding} based on a reduction from planar embedding to path-outerplanarity. The reduction structure is based on the one presented in \cite[Section 3.2]{FeuilloleyFMRRT21} for the planarity problem. We go over its details (adapted to the planar embedding problem). Given a graph $G=(V,E)$, a spanning tree $T$ rooted at some node $r$, and clockwise orderings $\rho(G)$, the reduction constructs a graph $h(G,T,\rho(G))$ which is composed of a path $P(G,T,\rho(G))$ and a set $Q(G,T,\rho(G))$ of edges between non-consecutive path nodes. 
	
	We first describe the construction of $P(G,T,\rho(G))$. For every node $v\in V$, let $\deg_{T}(v)$ and $\parent(v)$ denote $v$'s degree and parent in $T$, respectively. Let us denote by $\chi(v)$ the number of $v$'s children in  $T$, i.e., $\chi(v)=\deg_{T}(v)$ if $v=r$; and $\chi(v)=\deg_{T}(v)-1$ otherwise. If $v\neq r$, then for each $1\leq i\leq \chi(v)$, let $c_{i}(v)$ be $v$'s child for which $(v,c_{i}(v))$ is the $i$-th tree edge one encounters when following a clockwise ordering of the edges incident on $v$ starting from $(v,\parent(v))$. For $r$, we simply define $c_{i}(r)$ as $r$'s child for which the value $\rho_{r}(r,c_{i}(r))$ is the $i$-th smallest. Now, we can define the path $P(G,T,\rho(G))$ as the \emph{Euler tour} of $T$ starting from the root such that for each $v\in V$, its children are traversed in order $c_{1}(v),\dots ,c_{\chi(v)}(v)$. This means that for each $v\in V$, the path $P(G,T,\rho(G))$ contains $\chi(v)+1$ nodes $x_{0}(v),\dots,x_{\chi(v)}(v)$ and the path order is defined according to the following rules: (1) $x_{0}(r)$ is the leftmost node; (2) $x_{\chi(r)}(r)$ is the rightmost node; (3) for every non-leaf node $v\in V$ and $0\leq i<\chi(v)$, it holds that $x_{i}(v)$ is the left neighbor of $x_{0}(c_{i+1}(v))$; and (4) for every non-leaf node $v\in V$ and $0< i\leq \chi(v)$, it holds that $x_{i}(v)$ is the right neighbor of $x_{\chi(c_{i}(v))}(c_{i}(v))$.
	
	The set $Q(G,T,\rho(G))$ of non-path edges is defined based on the non-tree edges in $G$. For an edge $e=(u,v)\in E-T$, let $t(e,u)$ (resp., $t(e,v)$) be the first tree edge that one encounters when following a counterclockwise ordering with respect to $\rho_{u}$ (resp., $\rho_{v}$) starting from $e$. For every node $v\in V$ and non-tree edge $e\in E-T$ incident on $v$, we define the value $0\leq i(e,v)\leq \chi(v)$ as follows. If $t(e,v)=(v,\parent(v))$, then $i(e,v)=0$; otherwise $i(e,v)$ is defined as the index that satisfies $t(e,v)=(v,c_{i(e,v)}(v))$. The construction is completed by defining $Q(G,T,\rho(G))$ as the set of edges $\{(x_{i(e,u)}(u),x_{i(e,v)}(v))\mid e=(u,v)\in E-T\}$. Refer to Figure \ref{figure:reduction} for a pictorial example of the reduction. 
	
	\begin{figure}
		\centering
		\begin{subfigure}{0.3\textwidth}
			\includegraphics[width=1\linewidth]{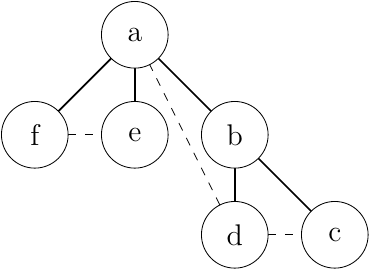}
			\caption{An embedded planar graph $G$ with a spanning tree $T$ (marked by the solid edges).}
			\label{fig:Ng1} 
		\end{subfigure}\vspace{0.4cm}
		\begin{subfigure}{0.7\textwidth}
			\includegraphics[width=1\linewidth]{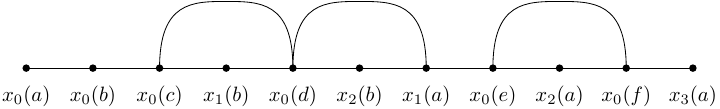}
			\caption{The graph $h(G,T,\rho(G))$.}
			\label{fig:Ng2}
		\end{subfigure}
		\caption{A reduction from planar embedding to path-outerplanarity.}
		\label{figure:reduction}
	\end{figure}


	The following lemma is established in \cite{FeuilloleyFMRRT21}.\footnote{Although the statement of Lemma \ref{lemma:reduction} itself does not explicitly appear in their paper, it can be derived from \cite[Proposition $4$]{FeuilloleyFMRRT21} and the constructive proof of \cite[Proposition $3$]{FeuilloleyFMRRT21}.}
	\begin{lemma}[\cite{FeuilloleyFMRRT21}]\label{lemma:reduction}
		The clockwise orderings $\rho(G)$ induce a planar embedding on $G$ if and only if $h(G,T,\rho(G))$ is a path-outerplanar graph w.r.t.\  $P(G,T,\rho(G))$ (i.e., the edges of $Q(G,T,\rho(G))$ are properly nested within $P(G,T,\rho(G))$).
	\end{lemma}  
	We are now prepared to prove Lemma \ref{lemma:path-outerplanarity-to-embedding}.
	\begin{proof}[Proof of Lemma \ref{lemma:path-outerplanarity-to-embedding}]
		The protocol starts with the prover computing a spanning tree $T$ of $G$ rooted at some node $r\in V$. Recall that the prover is able to send an encoding of $T$ to the verifier by means of the construction in Lemma \ref{lemma:tree-advice}. The idea is to have the verifier locally construct the graph $h(G,T,\rho(G))$ so that it is able simulate the given path-outerplanarity on $h(G,T,\rho(G))$ and make its decision accordingly. In parallel, the prover proves that $T$ is a spanning tree of $G$ in two interaction rounds by means of the protocol of Lemma \ref{lemma:npy-spanning-tree}. The soundness error of this spanning tree verification is reduced to $2^{-\ell}$ by means of an $\Theta(\ell)$-repetition.
		
		It is left to show that the protocol can be simulated on $h(G,T,\rho(G))$. Consider a node $v\in V$. We explain how $v$ is able to execute the path-outerplanarity protocol for all its copies $x_{0}(v),\dots x_{\chi(v)}(v)$ in $h(G,T,\rho(G))$. First, observe that given $T$ and $\rho_{v}$, the node $v$ can deduce the value $i(e,v)$ for all non-tree edges $e\in E-T$. Thus, $v$ is able to defer all edge-labels assigned to $e$ to the execution of node $x_{i(e,v)}(v)$ in $h(G,T,\rho(G))$. 
		
		As for the node-labels, throughout the protocol, the labels of $x_{i}(v)$ are assigned to $c_{i}(v)$ for each $1\leq i\leq \chi(v)$. Furthermore, $c_{i}(v)$ is responsible for the randomness of $x_{i}(v)$ and checks that the labels assigned to $x_{i}(v)$ throughout the protocol are consistent with its sampled bits. In addition, $v$ is assigned with the labels of the following nodes in $h(G,T,\rho(G))$: (1) $x_{0}(v)$; (2) the left neighbor of $x_{0}(v)$ in $P(G,T,\rho(G))$; (3) $x_{\chi(v)}(v)$; and (4) the right neighbor of $x_{\chi(v)}(v)$ in $P(G,T,\rho(G))$. Notice that by construction, each node $v$ can see the labels of all copies $x_{0}(v),\dots x_{\chi(v)}(v)$ as well as the labels of their left and right path-neighbors. To complete the simulation, the verifier executes consistency checks to verify that if two nodes in $G$ are given the labels of the same node in $P(G,T,\rho(G))$, then these labels are identical. Specifically, each node $v$ checks that the label designated to the left neighbor of $x_{0}(c_{i}(v))$ in $P(G,T,\rho(G))$ is identical to the label of $x_{i-1}(v)$ for all $1\leq i\leq \chi(v)$; and that the label designated to the right neighbor of $x_{\chi(c_{i}(v))}(c_{i}(v))$ in $P(G,T,\rho(G))$ is identical to the label of $x_{i}(v)$ for all $1\leq i\leq \chi(v)$. If the consistency checks succeed, then $v$ is able to simulate the verification at every node $x_{i}(v)$. 
		
		We now analyze the complexity and correctness of the protocol. First, note that by definition, the number of interaction rounds associated with simulating the path-outerplanarity protocol is $t$, whereas the number of interaction rounds associated with encoding $T$ and verifying that it is a spanning tree is $3$. So, we can conclude that the total number of interaction rounds is $\max\{t,3\}$. For the proof size, notice that each node $v\in V$ receives the labels of at most $5$ nodes in $h(G,T,\rho(G))$. Additionally, edge-labels can be simulated incurring only constant overhead due to Lemma \ref{lemma:edge-labels}. Therefore, the proof size of simulating the protocol is $O(\ell)$. This is added to the $O(\ell)$ proof size associated with the spanning tree verification on $T$. Regarding correctness, by Lemma \ref{lemma:reduction} and union bound arguments, the completeness error remains $\epsilon_{c}$ (since the spanning tree verification admits a completeness error of $0$) and the soundness error becomes $\epsilon_{s}+2^{-\ell}$ (since the spanning tree verification admits a soundness error of $2^{-\ell}$).
	\end{proof}
	
\section{Series-Parallel and Graphs of Treewidth at Most $2$}\label{section:series-parallel}
This section is devoted to designing protocols for series-parallel graphs and graphs of treewidth at most $2$ as stated in Theorems \ref{theorem:series-parallel} and \ref{theorem:treewidth}. 

We start from the protocol for series-parallel graphs. To that end, we shall use a characterization of series-parallel graphs which is based on the notion of a \emph{nested ear decomposition} presented in \cite{Eppstein92}. We note that this is not the ``common'' definition. Nevertheless, it will be convenient for our purposes. A nested ear decomposition of a graph $G=(V,E)$ is a partition of its edge-set $E$ into simple paths (referred to as ears) $P_{1},\dots ,P_{k}$ such that the following conditions hold: 
\begin{enumerate}
	\item The two endpoints of each ear $P_{j}\neq P_{1}$ lie in some ear $P_{i}$, $i<j$. Let us denote by $\mathcal{E}_{i}$ the set of ears for which the two endpoints lie in $P_{i}$.
	\item The interior nodes of $P_{j}$ do not appear in any ear $P_{i}$, $i<j$. 
	\item For every $i\geq 1$, the ears of $\mathcal{E}_{i}$ are properly nested within $P_{i}$. Put otherwise, the ears of $\mathcal{E}_{i}$ can be drawn above $P_{i}$ without crossings.
\end{enumerate}
The following equivalence is established in \cite{Eppstein92}.
\begin{lemma}[\cite{Eppstein92}]\label{lemma:nested-ear}
	A graph is series-parallel if and only if it admits a nested ear decomposition.
\end{lemma}
We are now prepared to describe the protocol.
\begin{proof}[Proof of Theorem \ref{theorem:series-parallel}]
	Suppose that $G=(V,E)$ is a series-parallel graph and let $P_{1},\dots, P_{k}$ be its nested ear decomposition. The idea is to have the prover encode the decomposition for the verifier and prove that it is indeed a nested ear decomposition. Let us define the paths $P'_{1},\dots, P'_{k}$ such that $P'_{1}=P_{1}$; and $P'_{i}$ is the subpath made up only of the interior nodes of $P_{i}$ for every $1<i\leq k$ ($P'_{i}$ is empty if $P_{i}$ is a single edge and $i>1$). Observe that by condition (2) in the definition of a nested ear decomposition, it holds that $P'_{1},\dots, P'_{k}$ partitions $V$ into node-disjoint simple paths. We refer to the paths $P'_{i}$ as \emph{sub-ears}. An edge $(u,v)$ is said to be \emph{$u$-connecting} if $u$ is an endpoint of the sub-ear $P'_{i}$ and $v$ is an endpoint of the ear $P_{i}$ for $i>1$. 
	
	At the beginning of the protocol, the prover provides the following information to the verifier: (i) an encoding of $F=\bigcup_{i\in [k]}P'_{i}$ based on the construction in Lemma \ref{lemma:tree-advice}; (ii) for each $v\in V$, the prover assigns a bit indicating if $v\in P_{1}$; and (iii) using the edge labels, the prover informs each endpoint $u$ of $P'_{i}$, $i>1$, which of its incident edges is $u$-connecting.
	
	For the rest of the protocol, we focus on the verification process within a single connected component $Q$ of $F$. That is, $Q$ is claimed by the prover to be a sub-ear. The full protocol is obtained by executing the described protocol on all components if $F$ in parallel. To verify that $Q$ is a simple path, each node $v\in Q$ first checks that its degree in $Q$ is at most $2$. Then, the prover and verifier execute the protocol of Lemma \ref{lemma:npy-spanning-tree} on $G(Q)$ to verify the connectivity and acyclity of $Q$. Additionally, if $Q$ is not marked by the prover as $P_{1}$, then each endpoint $u$ of $Q$ checks that it has exactly $1$ incident edge marked as $u$-connecting. 
	
	We note that if the described verification succeeds, then $Q$ is a simple path (up to a soundness error). Moreover, by construction, condition (2) of nested ear decompositions is satisfied. It is now left to explain how conditions (1) and (3) are verified. We note that as a byproduct, the verification of condition (3) will also verify that the connecting edges of $Q$ are node-disjoint (i.e., the ear that contains $Q$ is a simple path and not a cycle).
	
	To verify condition (1), we have the leftmost node of $Q$ uniformly sample a number $r_{Q}\in [\log ^{c}n]$ and send it to the prover. The prover receives the random values and assigns labels as follows. Consider a sub-ear $P'_{i}\neq P'_{1}$ and let $j<i$ be the index for which the endpoints of $P_{i}$ lie in $P_{j}$. The prover assigns each $v\in P'_{i}$ the pair $(\mathtt{ear}(v),\mathtt{pred\_ear}(v))=(r_{P'_{i}},r_{P'_{j}})$ (where $\mathtt{pred\_ear}(v)=\bot$ if $i=1$). In response, the verifier at every $v\in Q$ checks that it received the same label as its neighbors in $Q$. Additionally, if $v$ is an endpoint of $Q$ and $Q\neq P_{1}$, then it checks that $\mathtt{pred\_ear}(v)=\mathtt{ear}(u)$ where $(u,v)$ is the edge that was marked as $v$-connecting. This scheme verifies condition (1) for ears that are not a single edge. For ears that are a single edge $e=(u,v)$, the verifier at $u$ and $v$ simply check that their labels are identical.
	
	Finally, to obtain a protocol for condition (3), we would like to execute a protocol similar to the path-outerplanarity protocol of Theorem \ref{theorem:path-outerplanar} on each ear $P_{i}$ in parallel. The implementation idea is to treat the ears $P_{j}$ with both endpoints in $P_{i}$ as non-path edges in the path-outerplanarity protocol. This is done as follows. Suppose that $P_{j}$ is an ear which is not a single edge such that its endpoints are $u\in P_{i}$ and $v\in P_{i}$. Whenever the protocol requires the prover to assign label $L(u,v)$ to edge $(u,v)$, the prover instead assigns $L(u,v)$ to the nodes in $P'_{j}$. Then, each node $w\in P'_{j}$ checks that it received the same label as its neighbors in $P'_{j}$. Notice that this mechanism allows $u$ and $v$ to read $L(u,v)$ (from the labels of their respective neighbors in $P'_{j}$) and simulate the protocol as if $(u,v)$ is an edge. 
	
	In conclusion, in the described protocol the prover provides a nested ear decomposition and proves that it satisfies conditions (1)--(3). Performing the stages in parallel leads to $5$ interaction rounds and the proof size remains $O(\log\log n)$.
\end{proof}
We now move on to describing a protocol for graphs of treewidth at most $2$ as specified in Theorem \ref{theorem:treewidth}. 
The protocol is facilitated by the following known characterization. 
\begin{lemma}[\cite{Bodlaender98}]\label{lemma:tw2-sp-equivalence}
	A graph $G$ has treewidth at most $2$ if and only if every biconnected component of $G$ is series-parallel.
\end{lemma}
To design the protocol, we use a similar approach to the one used in the protocol for outerplanarity of Theorem \ref{theorem:outerplanar}. That is, we would like for the prover decompose the graph into its biconnected components and provide proof for the validity of the decomposition as well as that each component is series-parallel. To that end, we recall that a nested ear decomposition is a special case of \emph{open ear decompositions} and that a graph is biconnected if and only if it admits an open ear decomposition where the first ear is a single edge \cite{Whitney1931}. Moreover, it holds that if a graph is biconnected, then for any edge $e\in E$, there exists an open ear decomposition starting from $P_{1}=e$. We now turn to prove Theorem \ref{theorem:treewidth}.
\begin{proof}[Proof of Theorem \ref{theorem:treewidth}.]
	Consider a biconnected component $C$, and let $u$ be its $C$-separating node (as defined in Section \ref{section:outerplanar}. The prover defines the $C$-leader as some node $v\in C\cap N(u)$. Then, the prover computes a nested ear decomposition of $G(C)$ with $P_{1}=(u,v)$ as the first ear. 
	
	Following that, similarly to the outerplanarity protocol of Theorem \ref{theorem:outerplanar}, the prover seeks to encode the block-cut tree $T$ to the verifier and prove that: (1) every non-cut node is adjacent only to nodes in its biconnected component; (2) $T$ admits a tree structure; and (3) $G(C)$ is series-parallel for each biconnected component $C$. This is obtained based on the following modifications to the outerplanarity protocol. First, instead of a Hamiltonian path, the prover sends the nodes an encoding of a tree rooted at the $C$-leader and spans all nodes of $C$ apart from the $C$-separating node. Recall that such a tree exists as $G(C)$ remains connected after removing a single node. This allows the prover to prove (1) and (2) by similar measures to the ones presented in the outerplanarity protocol. For item (3), notice that by construction, the $C$-separating node $u$ is the leftmost node in each ear of $G(C)$ in which it participates. Thus, by similar arguments to the ones presented in the outerplanarity protocol, the prover and verifier can execute a protocol to verify that $G(C)$ admits a nested ear decomposition without assigning a label to $u$. Overall, the described construction yields a protocol for graphs of treewidth at most $2$.
\end{proof}
		\section{Lower Bound}
	In this section, we present a lower bound on the proof size of one-round protocols. In fact, the lower bound holds even if we augment the model with the following strengthening assumptions. First, assume that each node $v\in V$ receives a local input which consists of its own identifier as well as with its neighbors identifiers. Recall that in contrast, our upper bounds apply even if the nodes are \emph{anonymous}. Furthermore, the lower bounds hold even if we assume that the randomness in the protocol comes in the form of an unbounded random string \emph{shared} among the nodes.\footnote{Notice that shared randomness can only strengthen the model. This is because given a shared random string $r$, it is easy for each node $v\in V$ to simulate individual randomness simply by taking the bits of $r$ which are located at multiples of $\text{id}(v)$ as its random bits.} The full lower bound details are stated in the following lemma.
	
	\begin{lemma}\label{lemma:lb}
		Suppose that $\Pi$ is a family of planar graphs that contains every biconnected outerplanar graph. Then, any one-round protocol for deciding membership in $\Pi$ with completeness and soundness errors smaller than $1/10$ requires a proof size of $\Omega(\log n)$.
	\end{lemma}
	We remark that the lower bound applies even if we restrict the attention to instances with maximum degree $\Delta\leq 3$. We can now prove Theorem \ref{theorem:lb}.
	\begin{proof}[Proof of Theorem \ref{theorem:lb}]
		Lemma \ref{lemma:lb} immediately gives implies the desired lower bound for all problems stated in Theorem \ref{theorem:lb} apart from planar embedding. To handle the planar embedding, we recall that by Lemma \ref{lemma:embedding-to-planarity}, a protocol for planar embedding implies a protocol for planarity with the same number of interaction rounds and an additive overhead of $O(\log \Delta)$. Since the Lemma \ref{lemma:lb} applies even for graph with $\Delta\leq 3$, we can deduce an $\Omega(\log  n)$ lower bound for planar embedding which completes our proof.
	\end{proof}

	 We go on to prove the main lemma.
	\begin{proof}[Proof of Lemma \ref{lemma:lb}]
		The construction is based on the lower bound for proof labeling schemes presented in \cite{FeuilloleyFMRRT21} (we slightly modify the lower bound so that yes-instances will be biconnected). To adapt it to randomized protocols, we use a framework presented in \cite{FraigniaudMORT19}.
		
		The general idea of \cite{FeuilloleyFMRRT21} is to show that if the proof size is $o(\log n)$, then it is possible to define a set of yes-instances (i.e., biconnected outerplanar graphs) that when ``glued'' together produce a no-instance (i.e., a non-planar graph) such that the individual nodes cannot distinguish between the instances. 
		
		Let $n$ be an integer divisible by $4$ and let $a_{1},\dots ,a_{n},b_{1},\dots , b_{n}$ be a partition of $[n^{2}]$ into $2n$ disjoint sets of size $n/2$. For $a\in \{a_{1},\dots ,a_{n}\}$ and $b\in \{b_{1},\dots ,b_{n}\}$, define the yes-instance $G_{a,b}$ as follows. First, construct two disjoint paths $P_{a},P_{b}$ each with $n/2$ nodes. Let $P_{a}[j]$ (resp., $P_{b}[j]$) denote the $j$-th leftmost node in $P_{a}$(resp., $P_{b}$). The nodes of $P_{a}$ and $P_{b}$ are assigned with the IDs in $a$ and $b$, respectively, in increasing order (so that $P_{a}[j]$ and $P_{b}[j]$ are assigned the $j$-th smallest ID in $a$ and $b$, respectively). Let us denote $q_{1}=1,q_{2}=n/4,q_{3}=n/2$. For every $j\in \{1,2,3\}$, we add an edge between $P_{a}[q_{j}]$ and $P_{b}[q_{j}]$. Observe that the constructed graph is outerplanar. Moreover, the path $\langle P_{a}[1],P_{a}[2],\dots ,P_{a}[n/2],P_{b}[n/2],\dots ,P_{b}[2],P_{b}[1],P_{a}[1] \rangle$ forms a Hamiltonian cycle in the graph. Hence, the graph is biconnected outerplanar.
		
		Now, suppose that there is a one-round protocol with a proof size of $o(\log n)$ and for a given instance $G_{a,b}$, let $L_{a,b}$ denote the label assignment of the honest prover (i.e., the label assignment causing the verifier to accept the instance with probability larger than $9/10$). For an instance $G_{a,b}$, let us denote \[L(a,b)=(L_{a,b}(P_{a}[1]),L_{a,b}(P_{a}[n/4]),L_{a,b}(P_{a}[n/2]),L_{a,b}(P_{b}[1]),L_{a,b}(P_{b}[n/4]),L_{a,b}(P_{b}[n/2]))\ .\] Observe that $|L(a,b)|=o(\log n)$. Therefore, by a standard counting argument, there exist $a^{1},a^{2},a^{3}\in \{a_{1},\dots ,a_{n}\}$ and $b^{1},b^{2},b^{3}\in \{b_{1},\dots ,b_{n}\}$ such that $L(a^{i},b^{j})=L(a^{i'},b^{j'})$ for all $(i,i',j,j')\in \{1,2,3\}^{4}$. 
		
		The no-instance $G'$ is constructed as follows. Let $P_{a^{1}},P_{a^{2}},P_{a^{3}}$ (resp., $P_{b^{1}},P_{b^{2}},P_{b^{3}}$) be paths on $n/2$ nodes where the IDs in each $P_{a^{i}}$ (resp., $P_{b^{i}}$) are taken from $a^{i}$ (resp., $b^{i}$) in increasing order. Then, for every $1\leq i,j \leq 3$, add an edge between $P_{a^{i}}[q_{j}]$ and $P_{b^{i+j}}[q_{j}]$, where $i+j$ is taken modulo $3$ whenever larger than $3$. Observe that if one contracts the edges of each path $P\in \{P_{a^{1}},P_{a^{2}},P_{a^{3}},P_{b^{1}},P_{b^{2}},P_{b^{3}}\}$ into a single node, then we are left with the graph $K_{3,3}$. That is, $G'$ contains $K_{3,3}$ as a minor and thus, is not planar. 
		
		For the instance $G'$, define the label assignment $L'$ as follows. For each node $v\in P_{a^{i}}$ (resp., $v\in P_{b^{i}}$), assign the label $L'(v)=L_{a^{i},b^{i}}(v)$. The main observation that facilitates the lower bound is that given the label assignment $L'$, every node in the no-instance $G'$ has a local view which is identical (in distribution) to its local view in some yes-instance $G_{a^{i},b^{j}}$ given the label assignment $L_{a^{i},b^{j}}$. Let  $R_{i,j}$ denote the event that the verifier rejects the instance $G_{a^{i},b^{j}}$ given the label assignment $L_{a^{i},b^{j}}$ for each $1\leq i,j\leq 3$. Notice that by the completeness of the protocol, it follows that $\Pr[R_{i,j}]<1/10$. On the other hand, by the observation above, the probability of a node rejecting $G'$ is bounded from above by $\Pr[\bigcup_{(i,j)\in \{1,2,3\}^{2}}R_{i,j}]$. By a union bound argument, we can bound this probability by $\sum_{(i,j)\in \{1,2,3\}^{2}}\Pr[R_{i,j}]<9/10$. Since $G'$ is a no-instance, this contradicts the soundness of the protocol.
	\end{proof}

	\paragraph*{Acknowledgment:}
	We would like to thank Eylon Yogev for helpful discussions regarding \cite{NaorPY20}.
	
	\clearpage
	\bibliographystyle{alpha}
	
	\bibliography{references}
\end{document}